\documentclass[preprint]{elsarticle}

\usepackage{multirow,xspace,amsmath,amssymb,amsfonts,amsthm,algorithm,algorithmic,color}
\usepackage{framed}
\usepackage[shortcuts]{extdash}

\usepackage{titlesec}
\titleformat*{\section}{\Large\bfseries}
\titleformat*{\subsection}{\large\bfseries}

\newtheorem{theorem}{Theorem}[section]
\newtheorem{lemma}[theorem]{Lemma}
\newtheorem{corollary}[theorem]{Corollary}%

\theoremstyle{definition}
\newtheorem{definition}[theorem]{Definition}%
\newtheorem{remark}[theorem]{Remark}
\newtheorem{example}[theorem]{Example}

\newcommand\eat[1]{}

\usepackage{enumitem}
\setenumerate[1]{label=\rm(\it{\roman{*}}\rm),ref=({\it\roman{*}}),leftmargin=*}
\newlength{\wordlength}

\newcommand{\pref}{\succ\xspace}
\newcommand{\level}{\text{\sc Lev}}

\newcommand{\range}[2]{\in\{#1,\ldots,#2\}}

\newcommand{\allitems}{\mathcal{M}}
\newcommand{\allagents}{\mathcal{N}}

\newcommand{\uall}{\mathcal{U}}
\newcommand{\udd}{\mathcal{U^{DD}}}
\newcommand{\uid}{\mathcal{U^{ID}}}
\newcommand{\ubin}{\mathcal{U^{BIN}}}
\newcommand{\rev}[1]{#1_{\textrm{rev}}}

\begin{document}

\title{Fair Allocation with Diminishing Differences
\footnote{A preliminary version of this paper was published in the proceedings of IJCAI 2017 \citep{segal2017fair}. In the present version, the proofs are substantially revised and simplified, omitted proofs are included, and there are two new sections: one on the division of chores (section \ref{sec:increasing}) and one on binary utilities (section \ref{sec:binary}).}
}
	
\author{Erel Segal-Halevi} 
\ead{Ariel University, Ariel 40700, Israel. erelsgl@gmail.com}
\author{Avinatan Hassidim} 
\ead{Bar-Ilan University, Ramat-Gan 5290002, Israel. avinatanh@gmail.com}
\author{Haris Aziz} 
\ead{UNSW Sydney and Data61 CSIRO, Australia. haris.aziz@unsw.edu.au}

\begin{abstract}
Ranking alternatives is a natural way for humans to explain their preferences. It is used in many settings, such as school choice, course allocations and residency matches. 
Without having any information on the underlying cardinal utilities, arguing about the fairness of allocations requires extending the ordinal item ranking to ordinal bundle ranking. The most commonly used such extension is stochastic dominance (SD), where a bundle X is preferred over a bundle Y if its score is better according to all additive score functions. SD is a very conservative extension, by which few allocations are necessarily fair while many allocations are possibly fair.
We propose to make a natural assumption on the underlying cardinal utilities of the players, namely that the difference between two items at the top is larger than the difference between two items at the bottom. This assumption implies a preference extension which we call diminishing differences (DD), where X is preferred over Y if its score is better according to all additive score functions satisfying the DD assumption. We give a full characterization of allocations that are necessarily\-/proportional or possibly-proportional according to this assumption. Based on this characterization, we present a polynomial-time algorithm for finding a necessarily\-/DD-proportional allocation whenever it exists.
Using simulations, we compare the various fairness criteria in terms of their probability of existence, and their probability of being fair by the underlying cardinal valuations. We find that necessary\-/DD-proportionality fares good in both measures. 
We also consider envy-freeness and Pareto optimality under diminishing-differences, as well as chore allocation under the analogous condition --- increasing-differences.
\end{abstract}


\maketitle

\section{Introduction}
\label{sec:intro}
Algorithms for fair assignment of indivisible items differ in the kind of information they require from the users.

Some algorithms require the users to rank bundles of items, i.e., report a total order among the bundles. Examples are the Decreasing Demand procedure of \citet{Herreiner2002Simple}, the Approximate-CEEI procedure of \citet{budish2011combinatorial} and the two-agent  Undercut procedure of  \citet{Brams2012Undercut,Aziz2015Note}. The computational and communicational burden might be large, since the number of bundles is exponential in the number of items.

Other algorithms require the users to evaluate individual items, i.e., supply a numeric monetary value for each item. Such algorithms are often termed \emph{cardinal}. They often assume that the users' valuations are additive, so that the value of a bundle can be calculated by summing the values of the individual items. Examples are the Adjusted Winner procedure of \citet{Brams2000WinWin}, the approximate-maximin-share procedure of \citet{Procaccia2014Fair} and
the Maximum Nash Welfare procedure of \citet{Caragiannis2016Unreasonable}.
In this setting, the communication is linear in the number of items, but the mental burden may still be large, since assigning an exact monetary value to individual items is not easy. This is especially true when items are valued for personal reasons (such as when dividing inheritance) and do not have a market price.



This paper focuses on a third class of algorithms, which only require the users to rank individual items, i.e., report a total order among items. Such algorithms are often termed \emph{ordinal}.

Ordinal algorithms are ubiquitous in mechanism design. They are often used in real-world applications, such as the National Residency Matching Program \citep{Roth1997NRMP,ashlagi2014stability}), school choice applications \citep{abdulkadiroglu2003school}, and university admittance \citep{hassidim2016strategic,hassidim2016redesigning}. One reason for this is that it is relatively easy for people to state ordinal preferences. 
Another reason is that in some legacy systems, 
the input procedure asks for ordinal preferences only.
Often, the designer can change the allocation mechanism, but cannot change the input procedure, as agents do not want to learn new ways to enter their input into the system.

Ordinal algorithms are also common in AI and in fair division. Examples are the AL two-agent procedure of \citet{Brams2013TwoPerson}, the optimal-proportional procedure of \citet{Aziz2015Fair}, picking-sequence procedures~\citep{Brams2004Dividing,Bouveret2011General} and the envy-free procedures of \citet{Bouveret2010Fair}. Such algorithms often assume that the agents' preferences are implicitly represented by an additive utility function, which is not known to the algorithm. 
This creates ambiguity in the agents' bundle rankings. For example, if an agent ranks four items as $w\succ x \succ y\succ z$, then, based on additivity, the algorithm can know that e.g. $\{w,x\}\succ \{y,z\}$ and $\{w,y\}\succ \{x,z\}$, but cannot know the relation between $\{w,z\}$ and $\{x,y\}$. Algorithms cope with this problem in several ways.

1. \textbf{necessary\-/fairness criteria}. An allocation is called \textbf{necessarily\-/fair} if it is fair for \emph{all} additive utility profiles consistent with the reported item-rankings. Here, ``fair'' may be substituted by any fairness criterion, such as envy-freeness or proportionality, as well as Pareto\-/efficiency.
Necessary fairness is a strong requirement, which is not always satisfiable.
For example, the AL procedure finds a necessarily\-/envy-free allocation, but only for two agents, and even then, it might need to discard some of the items.

2. \textbf{possible\-/fairness criteria}. An allocation is called \textbf{possibly-fair} if it is fair for \emph{at least one} additive utility profile consistent with the reported item-rankings. Again, ``fair'' may be substituted by proportional or envy-free or Pareto\-/efficient. Possible fairness is a weak criterion; algorithms that only return possibly-fair allocations might be considered unfair by users whose actual utility function is different.

3. \textbf{Scoring rules}. A scoring rule is a function that maps the rank of an item to a numeric score. A common example is the Borda scoring rule \citep{Young1974Borda}, where the least desired item has a score of 1, the next item has a score of 2, and so on. The score of a bundle is the sum of the scores of its items. It is assumed that all agents have the same scoring function. I.e., even though agents may rank items differently, the mapping from the ranking to the numeric utility function is the same for all agents \citep{Bouveret2011General,Kalinowski2013Social,Baumeister2016Positional,Darmann2016Proportional}. This strong assumption weakens the fairness guarantee. The allocation may appear unfair to agents whose actual scoring rule is different.

\subsection{Contribution}
The present paper suggests an alternative between the strong guarantee of necessary\-/fairness and the weak guarantee of possible\-/fairness and scoring-rule-fairness.

We assume that people are more sensitive about which of their high-valued items they receive than about which of their low-valued items they receive.
Specifically, we assume that the utility-difference between the best item and the second-best item is at least as large as the utility between the second-best and the third-best, and so on. We call this assumption \emph{Diminishing Differences (DD)}. The DD assumption is satisfied by the Borda scoring rule, as well as by many other scoring rules, as well as by lexicographic preferences.

The DD assumption is supported by a survey that was recently reported by 
\citet{bronfman2015assigning}
in the context of matching medical students to hospitals for internships:
\begin{quote}
``The students were asked to fill surveys, to assert the difference between the first and the
second place, the second and the third place and so on.
Based on the surveys' results, more weight was given to the difference between first and second place than to the difference between the ninth and the tenth.''
\end{quote}

Based on the DD assumption, we formalize several fairness criteria. We call an allocation \textbf{necessarily\-/DD-fair} (NDD-fair) if it is fair according to all additive utility profiles satisfying the DD assumption, and \textbf{possibly-DD-fair} (PDD-fair) if it is fair according to at least one additive utility profile satisfying the DD assumption. Again, ``fair'' may be substituted by envy-free or proportional or Pareto\-/efficient. The following implications are obvious for any fairness criterion:

\begin{center}
necessarily\-/fair $\implies$ NDD-fair $\implies$
 PDD-fair $\implies$ Possibly-fair
\end{center}

In other words, the DD-fairness criteria are intermediate in strength between necessary\-/fairness and possible\-/fairness.
A formal definition of these criteria appears in {Section \ref{sec:dd}}.

The first question of interest is to decide, given an item ranking and two bundles, whether the NDD or the PDD relation holds between these bundles. We prove characterizations of the NDD and PDD set relations that provide linear-time algorithms for answering these questions.
Using these algorithms, it can be decided in polynomial time whether a given allocation is NDD-proportional or NDD-envy-free ({Section \ref{sec:NDD}}).

Next, we prove a necessary and sufficient condition for the existence of an NDD-proportional (NDDPR) allocation. Essentially, an NDDPR allocation exists if and only if it is possible to:

(a) give all agents the same number of items, and 

(b) give each agent his best item. 

\noindent
The proof is constructive and presents a simple linear-time algorithm for 
finding an NDDPR allocation whenever it exists
({Section \ref{sec:NDDPR}}).

To appreciate the difference between NDD-fairness and necessary\-/fairness, contrast the above condition with the so-called ``Condition D'' of \citet{Brams2013TwoPerson}, which is necessary and sufficient for the existence of a necessarily\-/proportional (NecPR) allocation for two agents. For NecPR, it is required that for every odd integer $k \in\{ 1,3,\ldots,2l-1\}$ (where the number of items is $2 l$), the agents have a different set of $k$ best items; the condition for NDDPR is ``Condition D'' limited to $k=1$.

Intuitively, NDDPR allocations are more likely to exist than NecPR allocations. On the flip side, an NDDPR allocation is more likely to be considered unfair by some agents (whose utility functions do not satisfy the DD assumption) than a NecPR allocation. 
To assess the magnitude of these two opposing effects, we conduct a simple simulation experiment. 
We find that the former effect is substantial: with  randomly-generated utility functions (with partially-correlated utilities), 
NDDPR allocations exist in between 20\% and 40\% more instances than NecPR allocations.
In contrast, the latter effect is much less substantial: 
when there are sufficiently many items, our simple algorithm for finding NDDPR allocations almost always yields an allocation that is proportional according to the cardinal utilities. This indicates that NDDPR is appealing as a normative fairness criterion ({Section \ref{sec:simulations}}).


While our main interest is in NDD-proportionality, we briefly present several extensions of our model.

First, instead of proportionality, we study the stronger property of \emph{envy-freeness} (EF). Since
every EF allocation is PR, every NDDEF allocation is NDDPR. Therefore, conditions (a) and (b) above are still necessary to NDDEF existence. When there are $n=2$ agents, EF is equivalent to PR, so NDDPR is equivalent to NDDEF and conditions (a) and (b) are also sufficient, and when they are satisfied, an NDDEF allocation can be found in linear time.
EF and PR diverge when there are three or more agents. 
When $n=3$, we show that an NDDEF allocation might not exist even if conditions (a) and (b) hold.
We then study the computational problem of deciding whether an NDDEF allocation exists. Since conditions (a) and (b) are necessary for NDDEF, the decision problem is trivial whenever the number of items is not an integer multiple of the number of agents, since then condition (a) is violated. It is also trivial if the number of items equals the number of agents, since in this case condition (b) is both necessary and sufficient. Therefore the first non-trivial case is when the number of items is twice the number of agents. We prove that the decision problem is NP-hard already in this case ({Section \ref{sec:NDDEF}}).


Second, we study Pareto\-/efficiency (PE).  
The DD assumption has a substantial effect on fairness criteria: NDD-fair allocations are easier (in terms of existence) than necessary\-/fair allocations and PDD-fair allocations are harder than possibly-fair allocations. Interestingly, the DD assumption does not have this effect on PE. We show that NDD-PE is equivalent to necessarily\-/PE and PDD-PE is equivalent to possibly-PE. So the DD assumption does not lead to a new efficiency criterion ({Section \ref{sec:efficiency}}).

Third, we study the allocation of \emph{chores} --- items with negative utilities. We assume that people care more about \emph{not} getting the \emph{worst} chore than about getting the best chore; this naturally leads to the condition of \emph{increasing differences (ID)}. 
While the basic definitions and lemmas for the DD relations have exact analogues for the ID relations, 
our characterization for existence of NDDPR allocation of goods has no direct analogue for NIDPR allocation of chores ({Appendix \ref{sec:increasing}}).

Finally, we compare the Diminishing-Differences assumption to another natural assumption which we call \emph{Binary}.
It is based on the assumption that each agent only cares about getting as many as possible of his $k$ best items, where $k$ is an integer that may be different for different agents. 
We show that, while the number of utility functions that satisfy this assumption (for a given preference relation) is much smaller than the number of DD utility functions, it does not lead to new fairness criteria: necessary\-/binary-fairness is equivalent to necessary\-/fairness and possible\-/binary-fairness is equivalent to possible\-/fairness ({Appendix \ref{sec:binary}}).

\subsection{Related Work}
Extending preferences over individual items
to sets of items is a natural and principled way of succinctly encoding preferences \citep{BBP04a}.
One of the most common set extensions is \emph{stochastic dominance} (SD).
It was developed for a different but related problem --- extending preferences over individual outcomes to \emph{lotteries} over outcomes. If $X,Y$ are lotteries, then $X\succsim^{SD} Y$ iff $E[u(X)]\geq E[u(Y)]$ for every weakly-increasing utility function $u$ \citep{HaRu69a,Bran17a}. 
In the context of fair item allocation, SD leads to the notions of necessary\-/fairness and possible\-/fairness \citep{Aziz2015Fair}. 
Other common extensions are \emph{downward-lexicographic} (DL) and \emph{upward-lexicographic} (UL) \citep{Cho2012Probabilistic,Bouveret2010Fair,NBR15a}.


The diminishing-differences extension, which is the focus of this paper, is quite natural but has not been formalized in prior work.
The most similar extension that we are aware of is the \emph{second-order stochastic dominance} (SSD). If $X,Y$ are lotteries, then $X\succsim^{SSD} Y$ iff $E[u(X)]\geq E[u(Y)]$ for every utility function $u$ which is weakly-increasing and \emph{weakly-concave} \citep{HaRu69a}. In the context of item assignment, weak concavity is equivalent to \emph{increasing} differences --- agents care more about \emph{not} getting the \emph{worst} item than about getting the best item. Increasing differences make sense in fair division of chores \citep{Aziz2017Chores}. We analyze this assumption in Appendix \ref{sec:increasing}.

\citet{bronfman2015redesigning,bronfman2015assigning} 
present a mechanism for matching students to hospitals for internships, where the students report rankings over the hospitals. 
Initially, using simulations of random-serial-dictatorship, each student is assigned a vector of probabilities for each hospital. 
To improve the efficiency of the random assignment, probabilities are traded between students with different rankings. 
Ensuring that each trade is mutually beneficial requires an assumption on the students' cardinal utilities. Based on the survey quoted in the introduction, it is assumed that the utility each student assigns to each hospital is the square of its Borda score, which is a special case of a DD utility function.

Besides fair division, set extensions have been applied for committee voting~\citep{ALL16a} and social choice correspondences~(see e.g., \citep{BDS01a,KePe84a}).
Recently, set extensions have also been used in philosophic works on ethics. Suppose an ethical agent has to choose between several actions. 
He/she is unsure between two ethical theories, each of which ranks the actions differently. 
Due to this uncertainty about theories, each action 
can be considered a lottery. 
Using the SD set extension, \citet{aboodi2017intertheoretic} and \citet{tarsney2018exceeding} show that, in some cases, the agent can choose an ethically-best action despite the ethical uncertainty.

In social choice theory, it is common to study restricted domains of preference profiles, such as single-peaked, single-crossing or level-$r$-consensus \citep{Mahajne2015Level,nitzan2018flexible}. Many problems are much easier to solve in such restricted domains than in the domain of all preferences \citep{Elkind2014Detecting,ELP17a}. The present paper focuses on a restriction to preferences satisfying the DD assumption, which has not been studied so far.

Many works on fair allocation of indivisible items look for allocations that are only approximately-fair, for example, envy-free up to at most one item \citep{lipton2004approximately,budish2011combinatorial}. In contrast, we are interested in allocations that are fair without approximations. Naturally, such allocations do not always exist, so we are interested in finding conditions under which they exist.

\section{Preliminaries}
\label{sec:model}
There is a set $\allagents{}$ of agents with $n=|\allagents|$. There is a set $\allitems$ of distinct items with $m=|\allitems|$.
A \emph{bundle} is a set of items. A \emph{multi-bundle} is a multi-set of items, i.e., it may contain several copies of the same item.\footnote{
Multi-bundles are used mainly as a technical tool during the proofs; our primary results concern simple bundles, that contain (at most) a single copy of each item.
}

An \emph{allocation} $\bf{X}$ is a function that assigns to each agent $i$  a bundle $X_i$, such that $\allitems = X_1 \cup \cdots \cup X_n$ and the $X_i$-s are pairwise-disjoint.

Each agent $i\in\allagents$ has a strict ranking $\pref_i$ on items.
Each agent may also have a utility function $u_i$ on (multi-)bundles.
When we deal with a single agent, we often omit the subscript $_i$ and consider an agent with ranking $\succ$ and utility-function $u$.

All utility functions considered in this paper are strictly positive and additive, so the utility of a (multi-)bundle is the sum of the utilities of the items in it.
A utility function $u$ is \emph{consistent with $\pref$} if for every two items $x,y$: 
\begin{align*}
u(\{x\}) > u(\{y\})
\iff
x\pref y
\end{align*}

We denote by $\mathcal{U}(\pref)$ the set of additive utility functions consistent with $\pref$.

Given a vector of $n$ rankings $\succ_1,\ldots,\succ_n$, we denote 
by $\mathcal{U}(\succ_1,\ldots,\succ_n)$ the set of vectors of additive utility functions $\mathbf{u} = (u_1,\ldots,u_n)$ such that for all $i\in\allagents$, $u_i$ is consistent with $\pref_i$.

The following definition is well-known (see for example
\citet{Aziz2015Fair}):
\begin{definition}
\label{def:sd}
Given a ranking $\succ$ and two (multi-)bundles $X,Y$:
\begin{align*}
X \succsim^{Nec} Y
&&
\iff
&&
\forall u\in \uall(\pref) \text{:~~~~~} u(X)\geq u(Y).
\\
X \succsim^{Pos} Y
&&
\iff
&&
\exists u\in \uall(\pref) \text{:~~~~~} u(X)\geq u(Y)
\end{align*}
\end{definition}

Given a strict ranking $\succ$, we assign to each item $x\in\allitems$ a \emph{level}, denoted $\level(x)$, such that the level of the best item is $m$, the level of the second-best item is $m-1$, etc. 
(this is also known as the \emph{Borda score} of the item).
We define the level of a multi-bundle as the sum of the levels of the items in it:
\begin{align*}
\level(X) := \sum_{x\in X} \level(x)
\end{align*}
where all copies of the same item have the same level.

In this work we assume that the agents truthfully report their rankings to the algorithm; we leave the issue of strategic manipulations to future work.

\section{The Diminishing-Differences Property}
\label{sec:dd}
We define our new concept of diminishing differences (DD) in three steps: first, we define the set of DD utility functions (Definition \ref{def:dd-utilities}). 
Based on this, we define the necessary\-/DD and possible\-/DD relations (Definition \ref{def:dd-relation}). 
Based on this, we define the NDD-fairness and PDD-fairness criteria (Definition \ref{def:dd-prop}).

\begin{definition}
\label{def:dd-utilities}
Let $\succ$ be a preference relation and
$u$ a utility function consistent with $\succ$.
We say that
$u$ has the \emph{Diminishing Differences (DD)} property if, for every three items with consecutive levels $x_3\succ x_2\succ x_1$ such that $\level(x_3)=\level(x_2)+1=\level(x_1)+2$,
it holds that $u(x_3)-u(x_2)\geq u(x_2)-u(x_1)$.

We denote by $\udd(\pref)$ the set of all DD utility functions consistent with $\pref$.

Given $n$ rankings $\pref_1,\ldots,\pref_n$, 
We denote by $\udd(\pref_1,\ldots,\pref_n)$ the set of all vectors of DD utility functions $\mathbf{u} = (u_1,\ldots,u_n)$, such that for all $i\in\allagents$, $u_i$ is consistent with $\pref_i$.
\end{definition}
The Borda utility function consistent with $\succ$ is a member of $\udd(\succ)$. Another member is the lexicographic utility function $Lex(x) := 2^{\level(x)}$, by which bundles are ordered by whether they contain the best item, then by whether they contain the second-best item, etc. 

An alternative characterization of $\udd$ is given by the following lemma.
\begin{lemma}
\label{lem:u-in-DD}
$u\in \udd(\pref)$ iff, for every four items $x_2, y_2, x_1, y_1$ with $x_2\succsim x_1$ and $y_2\succsim y_1$
and $x_2\neq y_2$ and $x_1\neq y_1$:
\begin{align}
\tag{*}
\frac{u(x_2)-u(y_2)}{\level(x_2)-\level(y_2)}
\geq
\frac{u(x_1)-u(y_1)}{\level(x_1)-\level(y_1)}
\end{align}
\end{lemma}
\begin{proof}
\textbf{DD $\implies$ (*)}:
Let $k:=|\level(x_2)-\level(y_2)|$. 
Then there are $k+1$ items whose level is between $x_2$ and $y_2$ (inclusive). Denote these items by $z_j$ for $j\in\{0,\ldots,k\}$, such that  $z_{k} \succ z_{k-1} \succ \ldots \succ z_0$, and either $z_k=x_2,z_0=y_2$ (if $x_2\succ y_2$) or vice versa: $z_k=y_2,z_0=x_2$ (if $y_2\succ x_2$).
Then, the left-hand side of (*) can be written as:
\begin{align*}
\frac{
u(z_k) - u(z_0)
}{
k
}
=
\frac{
\sum_{j=1}^{k} 
\big[
u(z_j) - u(z_{j-1})
\big]
}{
k
}
\end{align*}
This is an arithmetic mean of the $k$ differences $u(z_j) - u(z_{j-1})$, for $j\in\{1,\ldots,k\}$.

Similarly, let $k':=|\level(x_1)-\level(y_1)|$. 
The right-hand side of (*) is an arithmetic mean of $k'$ utility-differences of items with level between $x_1$ and $y_1$.

By assumption $x_2\succsim x_1$ and $y_2\succsim y_1$, so by DD, to each difference in the left-hand side corresponds a weakly-smaller difference in the right-hand side. Therefore, the arithmetic mean in the left-hand side is weakly larger.

\textbf{(*) $\implies$ DD}: in (*), let $y_2$ be the element ranked immediately below $x_2$, let $x_1=y_2$, and let $y_1$ be the element ranked immediately below $x_1$. Then the denominators both equal 1, and $u$ satisfies the DD definition.
\end{proof}

\begin{definition}
\label{def:dd-relation}
Given a ranking $\succ$ and two (multi-)bundles $X,Y$:
	\begin{align*}
	X \succsim^{NDD} Y
	&&
	\iff
	&&
	\forall u\in \udd(\pref) \text{:~~~~~} u(X)\geq u(Y)
	\\
	X \succsim^{PDD} Y
	&&
	\iff
	&&
	\exists u\in \udd(\pref) \text{:~~~~~} u(X)\geq u(Y)
	\end{align*}
\end{definition}

\begin{remark}
\label{rem:dd-sd}
Comparing Definitions \ref{def:sd} and \ref{def:dd-relation}, it is clear that:
\begin{align*}
X \succsim^{Nec} Y
\implies
X \succsim^{NDD} Y
\implies
X \succsim^{PDD} Y
\implies
X \succsim^{Pos} Y
\end{align*}
\end{remark}
We now define the main fairness criterion 
that we will investigate in this paper --- proportionality.
\begin{definition}
\label{def:dd-prop}
Given a vector $\mathbf{u}$ of utility functions, an allocation $\bf{X}$ is called \emph{proportional for $\mathbf{u}$} if $\forall i\in\allagents:~~ n\cdot u_i(X_i) ~\geq~ u_i(\allitems)$. 
\\
Given item rankings $\succ_1,\ldots,\succ_n$, an allocation $\bf{X}$ is called:
\begin{itemize}
\item  \emph{necessary\-/DD-proportional (NDDPR)} if it is proportional for all $\mathbf{u}\in \udd(\pref_1,\ldots,\pref_n)$.
\item \emph{possible\-/DD-proportional (PDDPR)} if it is proportional for at least one $\mathbf{u}\in \udd(\pref_1,\ldots,\pref_n)$.
\end{itemize}
For comparison, recall that an allocation $\bf{X}$ is called:
\begin{itemize}
\item  \emph{necessarily\-/proportional (NecPR)} if it is proportional for all $\mathbf{u}\in\mathcal{U}(\pref_1,\ldots,\pref_n)$.
\item \emph{possibly-proportional (PosPR)} if it is proportional for at least one $\mathbf{u}\in\mathcal{U}(\pref_1,\ldots,\pref_n)$.
\end{itemize}
\end{definition}
Like in Remark \ref{rem:dd-sd}, it is clear that
necessarily\-/proportionality implies NDD-proportionality implies PDD-proportionality implies possibly-proportionality.

We now give alternative characterizations of NDDPR and PDDPR in terms of the NDD and PDD relations. 
For every integer $k$ and bundle $X_i$, define $k\cdot X_i$ as the multi-bundle in which each item of $X_i$ is copied $k$ times. Proportionality can be defined by comparing, for each agent $i$, the bundle $X_i$ copied $n$ times, to the bundle of all items $\allitems$.
\begin{lemma}
\label{lem:dd-prop}
Given item rankings $\succ_1,\ldots,\succ_n$:

(a) An allocation $\bf{X}$ is NDDPR iff $\forall i\in\allagents:~ n\cdot X_i~\succsim_i^{NDD}~\allitems$.

(b) An allocation $\bf{X}$ is PDDPR iff
$\forall i\in\allagents:~ n\cdot X_i~\succsim_i^{PDD}~\allitems$.
\end{lemma}
\begin{proof}
Let $P(i,\mathbf{u})$ be the proportionality predicate ``$n\cdot u_i(X_i) ~\geq~ u_i(\allitems)$''.

(a)
The NDDPR definition is
``For all DD utility profiles $\mathbf{u}$, for all agents $i$, $P(i,\mathbf{u})$.''
The right-hand side is ``
For all agents $i$, 
for all DD utility profiles $\mathbf{u}$, 
$P(i,\mathbf{u})$.''
These statements are logically equivalent for any predicate $P$.

(b) The PDDPR definition is ``There exists a DD utility profile $\mathbf{u}$ for which, for all agents $i$, $P(i,\mathbf{u})$.''
The right-hand side is: ``For all agents $i$, 
there exists a DD utility profiles $\mathbf{u}$ such that
$P(i,\mathbf{u})$.''

The former definition logically implies the latter (for any predicate $P$). 
It remains to prove that the latter implies the former. 
Indeed, suppose that for every agent $i\in\allagents$, there exists $u_i\in\udd(\succ_i)$ such that $u_i(n\cdot X_i)\geq u_i(\allitems)$. By additivity, $u_i(n\cdot X_i) = n\cdot  u_i(X_i)$, so for every $i$, $n \cdot u_i(X_i)\geq u_i(\allitems)$. Therefore the allocation $\bf{X}$ is proportional by the profile $(u_1,\ldots,u_n)\in \udd(\succ_1,\cdots,\succ_n)$.
\end{proof}

\section{Characterizing NDD and PDD Relations}
\label{sec:NDD}
As a first step in finding DD-fair allocations among many agents, 
we study the NDD and PDD relations for a single agent. We are given a preference relation  $\succ$ on items and two multi-bundles $X,Y$, and have to decide whether $X \succsim^{NDD} Y$ and/or $X \succsim^{PDD} Y$.

We begin by proving a convenient characterization of the NDD relation. For the characterization, 
we order the items in each multi-bundle by decreasing level, so  $X=\{x_{-1},\ldots, x_{-|X|}\}$ where $x_{-1}\succsim \ldots \succsim x_{-|X|}$
(the order between different copies of the same item is arbitrary).%
\footnote{We use negative indices so that the order of indices is the same as the order of levels.}
For each $k\leq |X|$ we define $X^{-k}$ as the $k$ best items in $X$, i.e., $X^{-k}:=\{x_{-1},\ldots,x_{-k}\}$.

\begin{theorem}
\label{thm:X-NDD-Y}
Given a ranking $\succ$ and two (multi-)bundles $X,Y$, 
$X\succsim^{NDD} Y$ if and only if both of the following conditions hold:
\begin{enumerate}
\item $|X|\geq |Y|$ and
\item for each $k\in \{1,\ldots, |Y|\}$:
$\level(X^{-k}) \geq  \level(Y^{-k})$.
\end{enumerate}
\end{theorem}
Theorem \ref{thm:X-NDD-Y} implies that there is a polynomial-time algorithm to check whether $X\succsim^{NDD} Y$; 
see Algorithm~\ref{algo:comparedd}.

\begin{algorithm}[h]
\caption{Checking the $\succsim^{NDD}$ relation}
\label{algo:comparedd}
\centering
\renewcommand{\algorithmicrequire}{{\textbf{Input}:}}
\renewcommand{\algorithmicensure}{{\textbf{Output}:}}
 \begin{algorithmic}
\REQUIRE $X,Y\subset \allitems$, and a ranking $\succ$ of the items in $\allitems$.
\ENSURE  Yes if $X\succsim^{NDD} Y$; No otherwise.
 \end{algorithmic}
\algsetup{linenodelimiter=\,}
\begin{algorithmic}
	\IF{$|X|<|Y|$}
	\RETURN No  
	~~~~~~
	\COMMENT{condition (i) is violated}
	\ENDIF
	\STATE  Order the items in $X$ and $Y$ by decreasing order of preference, such that $x_{-1} \succsim \cdots \succsim  x_{-|X|}$ and $y_{-1} \succsim \cdots \succsim y_{-|Y|}$.
\STATE Initialize TotalLevelDiff$:= 0$.
\FOR{$j=1,\ldots,|Y|$}
\STATE LevelDiff := $[\level(x_{-j}) - \level(y_{-j})]$
\STATE TotalLevelDiff += LevelDiff
\IF{TotalLevelDiff $< 0$}
\RETURN No
	~~~~~~
	\COMMENT{condition (ii) is violated}
\ENDIF
\ENDFOR
 \RETURN Yes
\end{algorithmic}
\label{algo:main}
\end{algorithm}

\begin{remark}
\label{thm:X-Nec-Y}
Contrast this characterization with the following characterization of $\succsim^{Nec}$ from  \citet{Aziz2015Fair}. $X\succsim^{Nec} Y$ iff:
\begin{enumerate}
\item $|X|\geq |Y|$ and
\item for each $k\in \{1,\ldots, |Y|\}$:
$\level(x_{-k}) \geq \level(y_{-k})$.
\end{enumerate}
\end{remark}

Before proving Theorem \ref{thm:X-NDD-Y}, we give some examples.

\begin{example}
Suppose the set of items is $\allitems=\{1,\ldots,8\}$
and we are given a preference-relation $8\succ\cdots \succ 1$, so that each item is represented by its level.
Consider the following two bundles:
\begin{center}
$~~~~~X = \{8,4,2\}~~~~~$$~~~~~Y = \{7,6\}~~~~~$
\end{center}
Note that $|X|>|Y|$, $X$ is lexicographically-better than $Y$, and even the Borda score of $X$ is higher. However, the level of $X^{-2}$ (the two best items in $X$) is only $12$ while the level of $Y^{-2}$ is $13$. Hence, by Theorem \ref{thm:X-NDD-Y},  $X\not \succsim^{NDD} Y$. Indeed, $X$ is not better than $Y$ according to the DD utility function $u_{square}(x) := \level(x)^2$, since $u_{square}(X) = 84 < 85 = u_{square}(Y)$. 
\end{example}
\begin{example}
Consider the following two bundles:
\begin{center}
$~~~~~Z = \{8,5\}~~~~~$$~~~~~Y = \{7,6\}~~~~~$
\end{center}
Now the conditions of Theorem \ref{thm:X-NDD-Y} are satisfied: $|Z|\geq |Y|$ and $\level(Z^{-1})\geq \level(Y^{-1})$ and $\level(Z^{-2})\geq \level(Y^{-2})$. Hence the theorem implies that $Z\succsim^{NDD} Y$.
In contrast, condition (ii) in Remark \ref{thm:X-Nec-Y} is not satisfied since $\level(z_{-2}) < \level(y_{-2})$. 
Therefore $Z\not \succsim^{Nec} Y$.
Indeed, $Z$ is worse than $Y$ by some non-DD utility functions, for example, by $u_{sqrt}(x):=\sqrt{\level(x)}$,
since $u_{sqrt}(Z)\approx 5.06 < 5.09 \approx u_{sqrt}(Y)$.
\end{example}

\begin{proof}[Proof of Theorem \ref{thm:X-NDD-Y}]
~

\textbf{NDD $\implies$ (i) and (ii)}:
We assume that either \emph{(i)} or \emph{(ii)} is violated and 
prove that $X\not\succsim^{NDD} Y$, i.e., there is a utility function $u \in \udd(\succ)$ such that $u(X) < u(Y)$.

\begin{enumerate}
\item If \emph{(i)} is violated then $|Y|>|X|$. Define $u$ as:
\begin{align*}
u(z) := m|Y| + \level(z) && \text{~for~all~} z\in\allitems
\end{align*}
It has diminishing-differences since the difference in utilities between items with adjacent ranks is 1. 

The term $m|Y|$ is so large that the utility of a bundle is dominated by its cardinality. Formally, for every item $x$, $m|Y| < u(x) \leq m+m|Y|$, so:
 	\begin{align*}
 	u(X) &\leq |X|\cdot (m+m|Y|)
 	\\
 		   &<   m|Y| + |X|\cdot m|Y| && \text{since $|X|<|Y|$}
 	\\
 		   &=  (|X|+1)\cdot m|Y|
 	\\
 		   &\leq  |Y|\cdot m|Y|       && \text{since $|X|<|Y|$}
 	\\
 		   &<  u(Y)
 	\end{align*}
Hence $X\not\succsim^{NDD} Y$.

\item If \emph{(ii)} is violated then for some $k\geq 1$, $\level(Y^{-k}) > \level(X^{-k})$. Let $k$ be the smallest integer that satisfies this condition; hence $y_{-k} \succ x_{-k}$.
Let $C := \level(x_{-k})-1$ and define $u$ as:
\begin{align*}
u(z) := 
\begin{cases}
\level(z) & \text{for $z\prec x_{-k}$}
\\
[\level(z) - C]\cdot m|X| & \text{for $z\succsim x_{-k}$}
\end{cases}
\end{align*}
so the utilities of the items worse than $x_{-k}$ are $1, 2, \ldots, C$,
and the utilities of $x_{-k}$ and the items better than it are $m |X|, 2 m|X|, 3 m|X|, \ldots$.

This $u$ has diminishing-differences, since the difference in utilities between adjacent items ranked weakly above $x_{-k}$ is $m|X|$, the difference between $x_{-k}$ and the next-worse item is less than $m|X|$ and more than 1, and the difference between adjacent items ranked below $x_{-k}$ is 1. 

The term $m|X|$ is so large that the utility of a bundle is dominated by the level of its items that are weakly better than $x_{-k}$. Formally:
\begin{align*}
u(X) &= u(\{x_{-1},\ldots,x_{-k}\}) + u(\{x_{-(k+1)},\ldots,x_{-|X|}\})
\\
&=  m|X|\cdot [\level(\{x_{-1},\ldots,x_{-k}\}) - k\cdot C] 
+  \level(\{x_{-(k+1)},\ldots,x_{-|X|}\})
\end{align*}
The assumption $\level(Y^{-k})>\level(X^{-k})$ implies that 
$\level(X^{-k})\leq \level(Y^{-k}) -1$.
Hence 
the leftmost term is at most $m|X|\cdot [\level(\{y_{-1},\ldots,y_{-k}\}) - 1 - k\cdot C]$. Since the level of an item is at most $m$, the rightmost term is less than $m|X|$. Hence:
\begin{align*}
	u(X) &< m|X|\cdot [\level(\{y_{-1},\ldots,y_{-k}\}) - 1 - k\cdot C]
	       +  m|X|
	       \\
	       &=  m|X|\cdot [\level(\{y_{-1},\ldots,y_{-k}\}) - k\cdot C]
	       \\
	       & = u(Y^{-k}) \text{~~~~~~~~~~since~} y_{-1},\ldots,y_{-k}\succ x_{-k}
	       \\
	       &\leq u(Y).
\end{align*}
Hence $X\not\succsim^{NDD} Y$.
\end{enumerate}

\textbf{(i) and (ii)$\implies$NDD:}
~~~We assume that $|X|\geq |Y|$ and 
that $\forall k\range{1}{|Y|}: \level(X^{-k})\geq \level(Y^{-k})$.
We consider an arbitrary utility function $u\in\udd(\succ)$ and prove that
$\forall k\range{1}{|Y|}: u(X^{-k})\geq u(Y^{-k})$. This will imply that $u(X)\geq u(Y)$, so that $X\succsim^{NDD} Y$.

During the proof, we assume that for every $j\range{1}{|Y|}$: $x_{-j}\neq y_{-j}$. 
This does not lose generality, since if for some $j$ we have $x_{-j} = y_{-j}$, we can just remove this item from both $X$ and $Y$; this changes neither the assumptions nor the conclusion.

In the proof, we use the following notation.
\begin{itemize}
\item $l_{k} := \level(x_{-k}) - \level(y_{-k})$.
\item $L^{k} := \level(X^{-k}) - \level(Y^{-k}) = \sum_{j=1}^k l_{k}$.
\item $u_{k} := u(x_{-k}) - u(y_{-k})$.
\item $r_{k} := u_{k} / l_{k}$.
\item $U^{k} := u(X^{-k}) - u(Y^{-k}) = \sum_{j=1}^k u_{k} = \sum_{j=1}^k r_{k} l_{k}$.
\end{itemize}
In this notation, our assumptions are that
$\forall k\range{1}{|Y|}: l_k\neq 0 \text{ and } L^{k}\geq 0$. We have to prove that $\forall k\range{1}{|Y|}: U^{k}\geq 0$.

Suppose we walk on the graph of $L^{k}$ (see Figure \ref{fig:X-NDD-Y}). 
When we move from $L^{j-1}$ to $L^{j}$, we make $l_j$ steps (upwards if $l_j> 0$ or downwards if $l_j<0$).
By assumption, the graph is always above zero.
Hence, 
an earlier upwards step corresponds to every downwards step.

Suppose we walk simultaneously on the graph of $U^{k}$. 
When we move from $U^{j-1}$ to $U^{j}$, we make a step of size $u_j = r_j l_j$, or equivalently, $l_j$ steps of size $r_j$ (upwards if $l_j> 0$ or downwards if $l_j<0$). Hence, to every step of size $1$ on the graph of $L^k$ corresponds a step of size $r_j$ on the graph of $U^k$ (see Figure \ref{fig:X-NDD-Y}).
\begin{figure}
\begin{center}
\includegraphics[width=12cm]{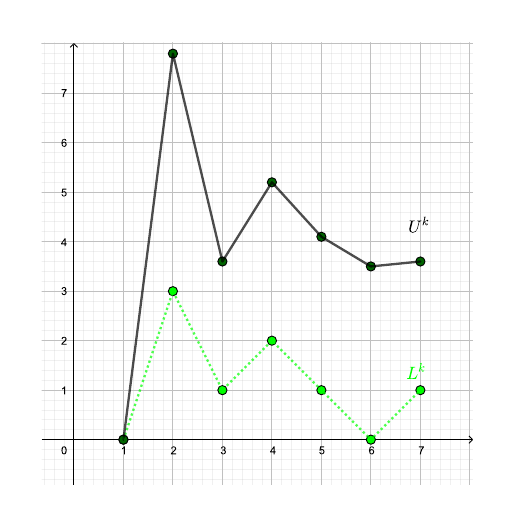}
\end{center}
\vskip -1cm
\caption{
\label{fig:X-NDD-Y}
An illustration of the graphs of $L^k$ and $U^k$ in the proof of Theorem \ref{thm:X-NDD-Y}.
}
\end{figure}

Now, we claim that $r_k$ is a weakly-decreasing function of $k$. Particularly, we claim that $i<j$ implies $r_i\geq r_j$.
To prove the claim we apply Lemma \ref{lem:u-in-DD}.
Since $u\in\udd(\succ)$, the lemma is applicable to $u$.
Since $i<j$, we have 
$x_{-i} \succsim x_{-j}$
and
$y_{-i} \succsim y_{-j}$.
By assumption, we have
$x_{-j} \neq y_{-j}$
and
$x_{-i} \neq y_{-i}$.
Therefore, the lemma implies:
\begin{align*}
&
\frac{u(x_{-i})-u(y_{-i})}{\level(x_{-i})-\level(y_{-i})}
\geq
\frac{u(x_{-j})-u(y_{-j})}{\level(x_{-j})-\level(y_{-j})}
\\
\iff
&
u_i / l_i \geq u_j / l_j
\\
\iff
&
r_i \geq r_j.
\end{align*}
Hence, to every step downwards of size $r_j$ on the graph of $U^k$ corresponds an earlier step upwards, and its size is at least $r_j$. 

Therefore, the graph of $U^{k}$, too, always remains above 0.
\end{proof}

Our next theorem gives an analogous characterization of the PDD relation.
\begin{theorem}
\label{thm:X-PDD-Y}
Given a ranking $\succ$ and two (multi-)bundles $X,Y$, 
$Y\succsim^{PDD} X$ if and only if at least one of the following conditions hold:
\begin{enumerate}
\item $|Y| > |X|$, or
\item for some $k\in \{1,\ldots, |Y|\}$:
$\level(Y^{-k}) > \level(X^{-k})$, or
\item $\level(Y) \geq \level(X)$.
\end{enumerate}
\end{theorem}

\begin{proof}
~

\textbf{(i) or (ii) or (iii)~~$\implies$~~PDD}:
~~~If (i) holds, then $u(Y)\geq u(X)$ by the DD function $u(z)$ in the proof of Theorem \ref{thm:X-NDD-Y}(i).
Similarly, if 
(ii) holds, then $u(Y)\geq u(X)$ by the DD function $u(z)$ in the proof of Theorem \ref{thm:X-NDD-Y}(ii).

If (iii) holds, then $u(Y)\geq u(X)$ by the DD function $u(z) := \level(z)$.

\textbf{PDD~~$\implies$~~(i) or (ii) or (iii)}:
We assume that none of the three conditions holds, and prove that $Y\not\succsim^{PDD} X$. So we have:

$\widehat{(i)}$ $|X|\geq |Y|$, and

$\widehat{(ii)}$ $\forall k\range{1}{|Y|}: \level(X^{-k})\geq \level(Y^{-k})$, and

$\widehat{(iii)}$ $\level(X) > \level(Y)$.

We consider an arbitrary function $u\in\udd(\succ)$, and show that $u(X) > u(Y)$. 
During the proof, we denote $K:=|Y|$.

We use the notation of the proof of Theorem \ref{thm:X-NDD-Y}.
By conditions $\widehat{(i)}$ and $\widehat{(ii)}$, the graph of $L^{k}$ is always weakly above zero. Hence, to every step downwards corresponds an earlier step upwards.
As in the proof of Theorem \ref{thm:X-NDD-Y}, the graph of $U^k$ is always weakly above zero, so $\forall k\range{1}{K}: u(X^{-k})\geq u(Y^{-k})$.
Now we consider two cases.

\emph{Case \#1}: the graph of $L^k$ ends strictly above zero. 
Hence, there exists a step upwards with no corresponding step downwards. 
Therefore the graph of $U^{k}$, too, ends  strictly above zero.
Therefore, we have $u(X^{-K}) > u(Y^{-K})$.
Since $u(X) \geq u(X^{-K})$ and $Y^{-K} = Y$, we get $u(X) > u(Y)$.

\emph{Case \#2}: 
the graph of $L^k$ ends at zero. 
So we have $\level(X^{-K})=\level(Y^{-K})=\level(Y)$.
Now, $\widehat{(iii)}$  says that 
$\level(X)>\level(Y)$; this means that $X$ must contain items that are not in $X^{-K}$.
We assume that utilities are strictly positive, so $u(X) > u(X^{-K})$. Since $u(X^{-K}) \geq u(Y^{-K})$ and $Y^{-K} = Y$, we get $u(X) > u(Y)$.

The same is true for every $u\in\udd(\succ)$. Hence $Y\not\succsim^{PDD} X$.
\end{proof}

Theorem \ref{thm:X-PDD-Y} implies that there is a polynomial-time algorithm to check whether $X\succsim^{PDD} Y$; 
the algorithm is similar to Algorithm~\ref{algo:comparedd} and we omit it.

Using Theorem \ref{thm:X-PDD-Y}, we illustrate the difference between PDD-fairness and possible\-/fairness.
\begin{example}[PDD-fairness vs. possible\-/fairness]
\label{exm:pdd}
There are $m = 2 l$ items, for some $l\geq 3$.
Alice and Bob have the same preferences:
\begin{align*}
2 l \succ 2 l - 1 \succ ... \succ 4 \succ 3 \succ 2 \succ 1
\end{align*}
Both Alice and Bob get $l$ items: Alice gets $2 l, 2 l - 1, ...l + 3, l + 2, 1$ and Bob  gets $l + 1,l,...3,2$. Intuitively this allocation seems very unfair since Alice gets all the $l-1$ best items.  However, it is possibly-proportional, since Bob's utility function might assign the a value near 0 to item 1 and a value near 1 to all other items.

In better accordance with our intuition, the above allocation is not PDD-proportional: by Theorem \ref{thm:X-PDD-Y}, Bob's bundle is not PDD-better than Alice's bundle, since it does not satisfy any of the conditions (i) to (iii). \qed
\end{example}

Based on the two constructive theorems proved in this section, we have:
\begin{corollary}
\label{cor:polynomial}
The following problems can be decided in polynomial time:

(a) Given an allocation, decide whether it is NDDPR;

(b) Given an allocation, decide whether it is PDDPR.
\end{corollary}


\section{Existence of NDD-Proportional Allocations}
\label{sec:NDDPR}
In this section, we prove a necessary and sufficient condition for the existence of NDDPR allocations.
\begin{theorem}
	\label{thm:nddpr}
	An NDDPR allocation exists if and only if:
	\begin{itemize}
		\item[] (a) The number of items is a multiple of the number of agents, i.e., $m=l\cdot n$, where $l$ is an integer and $n$ is the number of agents, and
		\item[] (b) Each agent has a different best item.
	\end{itemize}
	In case it exists, it can be found in time $O(m)$.
\end{theorem}

\begin{proof}
~

\textbf{NDDPR $\implies$ (a) and (b):}
Let $X_1,\ldots,X_n$ be an NDDPR allocation. So for all $i\in\allagents$, ~~$n\cdot X_i \succsim^{NDD}_i \allitems$. By the two conditions of Theorem \ref{thm:X-NDD-Y}:

(a) For all $i\in\allagents$: $|n\cdot X_i| \geq  |\allitems|
\implies n\cdot |X_i|  \geq  m.$
But this must be an equality since the total number of items in all $n$ bundles is exactly $m$. Therefore, the total number of items is $n\cdot |X_i|$ which is an integer multiple of $n$.

(b) For all $i$, the level of the best item in $n\cdot X_i$ must be weakly larger than the level of the best item in $\allitems$. So for all $i\in\allagents$, $X_i$ must contain agent $i$'s best item. So the best items of all agents must be different.

\textbf{(a) and (b)$\implies$ NDDPR:}
We show that, if (a) and (b) hold, then the \emph{balanced round-robin} algorithm (Algorithm \ref{algo:roundrobin}) produces an NDDPR allocation.%
\begin{algorithm}[h]
\caption{Balanced round-robin allocation of items}
\label{algo:roundrobin}
\centering
\algsetup{linenodelimiter=\,}
\begin{algorithmic}
\WHILE{there are remaining items}
\FOR{$i=1,\ldots,n$}
\STATE Give agent $i$ his best remaining item.
\ENDFOR
\FOR{$i=n,\ldots,1$}
\STATE Give agent $i$ his best remaining item.
\ENDFOR
\ENDWHILE
\end{algorithmic}
\end{algorithm}

\noindent
Let $X_i$ be the bundle allocated to agent $i$ by balanced-round-robin. 
We prove that $n\cdot X_i \succsim^{NDD}_i \allitems$ by the two conditions of Theorem \ref{thm:X-NDD-Y}.

Condition (i) is satisfied with equality, since by (a) each agent gets exactly $l$ items, so $|n\cdot X_i| = n l = m = |\allitems|$.

Condition (ii) says that, for every $k\range{1}{m}$, the total level of the $k$ best items in the multi-bundle $n\cdot X_i$ should be at least as large as the total level of the $k$ best items in $\allitems$.
It is convenient to verify this condition following Algorithm \ref{algo:comparedd}: we have to prove that, when going over the items in both bundles from best to worst, the total level-difference between them (the variable TotalLevelDiff in the algorithm) remains at least 0.

We first prove that this is true after the first round. By condition (b), in the first round, each agent receives his best item, so the level of the best $n$ items in $n\cdot X_i$ is $m$. 
The following table shows the levels and their differences for $k\in\{1,\ldots,n\}$ (here, it is important that all items in $\allitems$ are distinct):
\begin{center}
\begin{tabular}{llllll}
	 	& $k=1$ & $k=2$ & $k=3$ & $\ldots$ & $k=n$ \\
	$n\cdot X_i$	& $m$ & $m$ & $m$ & $\ldots$ & $m$ \\
	$\allitems$		& $m$ & $m-1$ & $m-2$ & \ldots & $m-n+1$ \\
	LevelDiff		& $0$ & $1$ & $2$ & \ldots & $n-1$ \\
	TotalLevelDiff	& $0$ & $1$ & $3$ & \ldots & $n(n-1)/2$ \\
\end{tabular}
\end{center}
We now prove that, after each round $r\geq 1$, the accumulated level-difference TotalLevelDiff for agent $i$ is at least $n(n-1)/2$ when $r$ is odd, and at least $n(i-1)$ when $r$ is even. We also prove that TotalLevelDiff is always at least 0.

The proof is by induction on $r$. We have just proved the base $r=1$.

Suppose now that $r>1$ and $r$ is even.
When agent $i$ picks an item, 
the number of items already taken is $r n - i$.
Therefore, agent $i$'s best remaining item has a level of at least $m-(r n-i)$. 
Therefore, the level-differences 
for $k\in\{(r-1)n+1,~\ldots~,r n\}$ are as in the following table (where the last row uses the accumulated level-difference of $n(n-1)/2$ from the induction assumption):
\begin{center}
	\small
	\begin{tabular}{lllll}
		$n\cdot X_i$	& $\geq m-r n+i$ & $\geq m-r n+i$ & $\ldots$ & $\geq m-r n+i$ \\
		$\allitems$		& $m-r n + n$ & $m-r n + n-1$ & \ldots & $m-r n + 1$ \\
		LevelDiff		& $\geq i-n$ & $ \geq i-n+1$ & \ldots & $\geq i-1$ \\
		TotalLevelDiff	& $\geq {n(n-1)\over 2}+i-n$ & $ \geq {n(n-1)\over 2}+2 i-2 n+1$ & \ldots & $\geq n(i-1)$ \\
	\end{tabular}
\end{center}
The sum of terms in the LevelDiff row is  $\frac{n[(i-n)+(i-1)]}{2} = \frac{n[- n -1]}{2} + n i$.
Adding the $\frac{n(n-1)}{2}$ from the induction assumption gives that, at the round end, TotalLevelDiff is at least $ni - n = n(i-1)$ as claimed.
We now show that TotalLevelDiff is at least $0$ throughout the round.
LevelDiff is non-positive in the first $n-i+1$ columns of the table, and positive afterwards. So TotalLevelDiff attains its smallest value at step $n-i+1$.
The sum of LevelDiff from step $1$ to step $n-i+1$ is $(i-n)(n-i+1)/2$. Hence TotalLevelDiff at step $n-i$ is at least 
$n(n-1)/2 - (n-i+1)(n-i)/2$.
Since $n\geq n -i+1$ and $n-1\geq n-i$, this expression is at least 0.

%
%

Suppose now that $r>1$ and $r$ is odd.
When agent $i$ gets an item, 
the number of items already taken is $r n - (n-i+1)$. Therefore, agent $i$'s best remaining item has a level of at least $m- r n + (n-i+1)$. Therefore, the level-differences 
for $k\in\{(r-1)n+1,~\ldots~,r n\}$ are as in the following table:
\begin{center}
	\scriptsize
	\begin{tabular}{lllll}
		$n\cdot X_i$	& $\geq m-r n+n-i+1$ & $\geq m-r n+n-i+1$ & $\ldots$ & $\geq m-r n+n-i+1$ \\
		$\allitems$		& $m-rn+n$ & $m-rn+n-1$ & \ldots & $m-rn+1$ \\
		LevelDiff		& $\geq 1-i$ & $\geq 2-i$ & \ldots & $\geq n-i$ \\
		TotalLevelDiff		& $\geq n(i-1)+1-i$ & $\geq n(i-1)+3-2 i$ & \ldots & $\geq n(n-1)/2$ \\
	\end{tabular}
\end{center}
The sum of terms in the LevelDiff row is $\frac{n[(1-i)+(n-i)]}{2} = \frac{n[n+1]}{2} - n i$. Adding the $n(i-1)$ from above gives that, at the round end, TotalLevelDiff is at least $n(n-1)/2$ as claimed.
We now show that TotalLevelDiff is at least $0$ throughout the round.
LevelDiff is non-positive 
in the first $i$ 
columns of the table, and positive afterwards. So TotalLevelDiff attains its smallest value at step $i$.  
The sum of LevelDiff from step $1$ to step $i$ is $i(1-i)/2$. Hence TotalLevelDiff at step $i$ is at least 
$n(i-1) + i(1-i)/2 = (i-1)(n-i/2)\geq 0$.
\end{proof}

Using Theorem \ref{thm:nddpr}, we illustrate the difference between NDD-fairness and necessary\-/fairness.
\begin{example}[NDD-fairness vs. necessary\-/fairness]
\label{exm:ndd}
Suppose the set of items is $\allitems=\{1,\ldots,2 l\}$, for some $l\geq 2$. Alice and Bob have almost opposite preferences:
\begin{align*}
\text{Alice:}&& 2 l \succ 2 l - 1 \succ ... \succ 4 \succ 3 \succ 2 \succ 1
\\
\text{Bob:}&&   2\succ 3 \succ 4 \succ... \succ 2 l - 1 \succ 2 l \succ 1
\end{align*}
Intuitively we would expect that opposite preferences make it easy to attain a fair division. 
However, in this case, no necessarily\-/proportional 
allocation exists:
By Remark \ref{thm:X-Nec-Y}, in a necessarily\-/fair allocation both agents must receive the same number of items ($l$). But Alice and Bob have the same worst item ($1$), so one of them must get it. Suppose it is Alice. So Alice has only $l-1$ items better than $1$, while Bob has $l$ items better than $1$. Hence, the allocation is not necessarily\-/proportional for Alice (her utility function might assign a value near 0 to this item and a value near 1 to all other items).

In contrast, our Theorem \ref{thm:nddpr} shows that an NDD-proportional allocation exists. Intuitively, since it is possible to give each agent his/her best items, they are willing to compromise on the less important items.\qed
\end{example}

\section{Simulation Experiments}
\label{sec:simulations}
\begin{figure*}[h]
\hskip -15mm
\includegraphics[width=150mm,height=70mm]{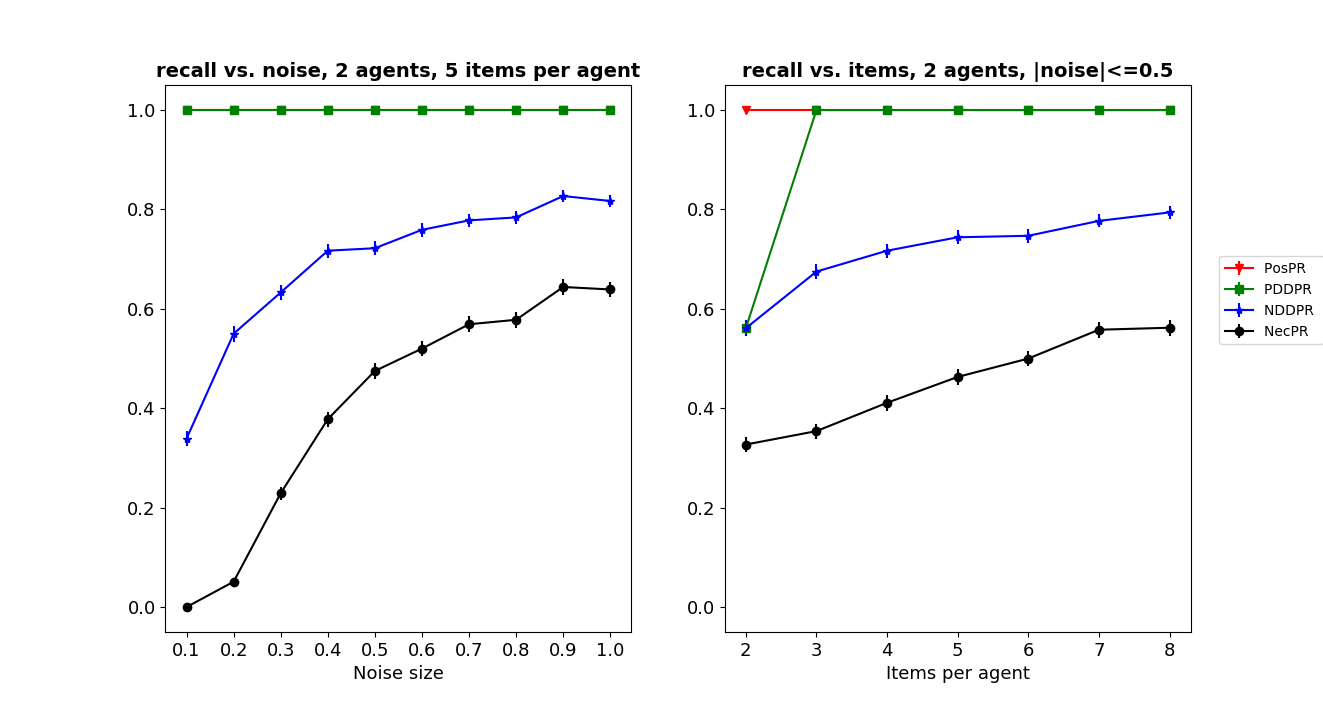}
\caption{Estimated ``recall'' --- the fraction of preference-profiles that admit an allocation satisfying each fairness criterion. 
Vertical bars denote sample standard error.
Lines connecting data-points are for eye-guidance only.
~~~
The top line corresponds to both PosPR and PDDPR --- for both of them the estimated recall is $1$, which means that all utility profiles we checked admit such allocations. The lines below them correspond to NDDPR and NecPR respectively.
\label{fig:recall}}
\end{figure*}
\begin{figure*}[h]
\hskip -15mm
\includegraphics[width=150mm,height=70mm]{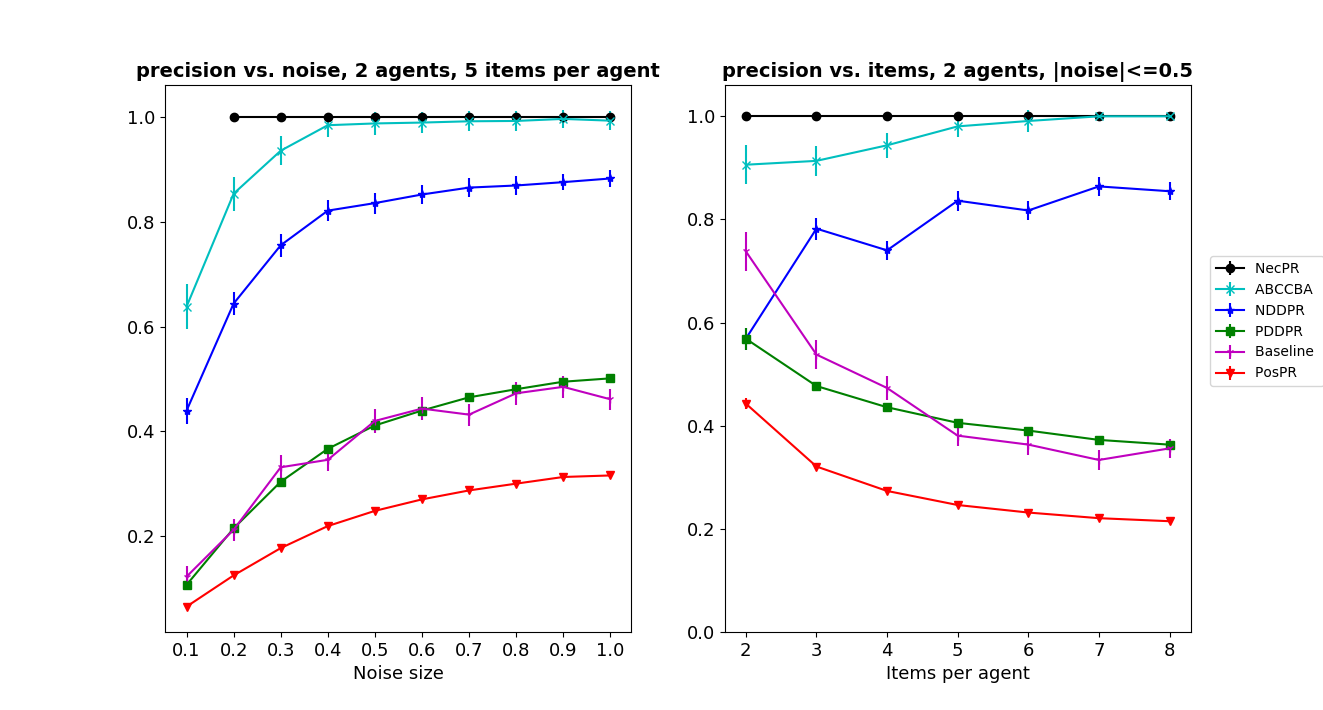}
\caption{Estimated ``precision'' --- the fraction of allocations that are fair according to the cardinal valuations, among those that are fair by the ordinal fairness criterion.
Vertical bars denote sample standard error.
Lines connecting data-points are for eye-guidance only.
~~~
The lines, from top to bottom, correspond to:
(a) NecPR --- by definition it is necessarily always $1$ (the point at noise size 0.1 is missing since no profile with this noise admitted a NecPR allocation);
(b) The NDDPR allocations found by the balanced-round-robin protocol (Algorithm \ref{algo:roundrobin});
(c) An arbitrary NDDPR allocation;
(d) An allocation found by a baseline protocol in which the first round is like Algorithm \ref{algo:roundrobin} but the following items are allocated at random;
(e) An arbitrary PDDPR allocation;
(f) An arbitrary PosPR allocation.
\label{fig:precision}}
\end{figure*}
A mechanism designer who has to choose a fairness criterion faces a tradeoff: choosing a weak criterion (such as PosPR or PDDPR) makes it easier to find an allocation that satisfies the criterion but also makes it more likely that some agents will consider it unfair. In contrast, with a strong criterion (such as NDDPR or NecPR), it is harder to find an allocation, but once an allocation is found, it is more likely that agents will consider it fair.
This tradeoff is analogous to the tradeoff between ``recall''  and ``precision'' in information retrieval and binary classification.%
\footnote{
See the Wikipedia page ``Precision and Recall'' for a definition of these terms in information retrieval and binary classification.
}
Given a fairness-criterion, we define its recall and precision as follows:
\begin{itemize}
\item The \emph{recall} of the criterion is the 
probability that a random utility-profile admits an allocation satisfying this criterion;
\item the \emph{precision} of a fairness-criterion is the 
probability that a random allocation satisfying this criterion according to the ordinal rankings is indeed fair according to the cardinal valuations.
\end{itemize}
We estimated the recall and precision of various fairness criteria as follows.

\subsection{Randomly-generated instances}
To simulate valuations with partial correlation, we determined for each item a ``market value'' drawn uniformly at random from $[1,2]$. We determined the cardinal value of each item to each agent as the item's market value plus noise drawn uniformly at random from $[-A,A]$, where $A\in[0,1]$ is a parameter.
Based on the cardinal values, we determined the agent's ordinal ranking. 
Then, for each such utility-profile, we checked various statistics:
\begin{itemize}
\item How many allocations are NecPR/NDDPR/PDDPR/PosPR according to the ordinal rankings;
\item How many NecPR/NDDPR/PDDPR/PosPR allocations are indeed proportional according to the underlying cardinal valuations;
\item Whether the specific NDDPR allocation found by the procedure of Theorem \ref{thm:nddpr} is proportional according to the underlying cardinal valuations;
\item As a baseline, we also checked the fairness of an allocation found (under the conditions of Theorem \ref{thm:nddpr}) by giving each agent its favorite item and dividing the remaining items randomly.
\end{itemize}
We did this experiment for $n\in\{2,3\}$ agents, for different values of $A\in \{0.1,\ldots,1\}$, and for different numbers $l$ of items per agent ---  $l\in\{2,\ldots,8\}$ when $n=2$ or $l\in\{2,\ldots,5\}$ when $n=3$. 
For each combination, we checked $1000$ randomly-generated instances.\footnote{
The Python code used for the experiments is available at GitHub:\\ https://github.com/erelsgl/fair-diminishing-differences
}

Below we report the results for $n=2$ agents; the results for $n=3$ agents are qualitatively similar and we omit them from the paper.\footnote{
All results and plots can be found online:
\\
https://github.com/erelsgl/fair-diminishing-differences/blob/master/results/Readme.md
}

\subsection{Results --- recall}
Figure \ref{fig:recall} presents the results for recall (probability of existence). As expected, 
the recall of the weak criteria --- PosPR and PDDPR --- is almost always 1; the recall of NDDPR is lower, but it is still significantly higher than that of NecPR. Thus, an NDDPR allocation is likely to exist in many cases in which a NecPR allocation does not exist.

As expected, both kinds of allocations are more likely to exist when the noise size $A$ is larger, since larger noise corresponds to less correlated rankings.
Similarly, both kinds of allocations are more likely to exist when there are more items to share; this finding resembles the results of  \citet{dickerson2014computational} for envy-free allocations with cardinal valuations.

\subsection{Results --- precision}
Figure \ref{fig:precision} presents the results for precision (probability of fairness). 
NecPR allocations, when they exist, are always proportional by definition; hence the precision of NecPR is always $1$.
The precision of NDDPR is lower than $1$, but it is much higher than that of the weaker criteria --- PosPR and PDDPR.

Interestingly, the specific NDDPR allocation found by the round-robin protocol of Theorem \ref{thm:nddpr} is very likely to be proportional --- in most cases its precision is very near $1$.

Note that, since the randomization we used is completely uniform and does not use the DD assumption, the probability that DD holds is very low.%
\footnote{
There are $n l$ items, so there are $n l - 1$ differences between utilities of adjacent items. DD requires that, for each agent, these differences be ordered in a descending order.
With high probability, all differences are distinct, so there are $(n l - 1)!$ different orders, and only one of them  corresponds to a DD utility function. Therefore, the probability that DD holds for each single agent is $1/(n l-1)!$, and for all $n$ agents it is $1/((n l-1)!)^n$. 
}
Nevertheless, the NDDPR allocation of Theorem \ref{thm:nddpr} (when it exists) is almost always proportional when the number of items or the noise size is sufficiently large. This further shows the robustness of our algorithm.%

Comparing the two graphs, we see that the NecPR requirement is too strong, and the PDDPR and PosPR requirements are too weak, while the NDDPR requirement hits a sweet spot between recall and precision: it allows us to solve a large fraction of the instances, and the solutions are likely to be considered fair by the agents.

\section{Envy-freeness}
\label{sec:NDDEF}
The following is an analogue of the definition of proportionality-related fairness criteria (Definition \ref{def:dd-prop}):
\begin{definition}[Envy-freeness]
\label{def:dd-ef}
Given utility functions $u_1,\ldots,u_n$, an allocation $\bf{X}$ is called \emph{envy-free (EF)} if $\forall i,j\in\allagents:~~ u_i(X_i) ~\geq~ u_i(X_j)$. 
~~~
Based on this definition, \emph{necessary\-/DD-envy-free (NDDEF)} and \emph{possible\-/DD-envy-free (PDDEF)} are defined analogously to NDDPR and PDDPR.
\\
\end{definition}

The following is a partial analogue of Lemma \ref{lem:dd-prop}
and contains an alternative characterization of NDDEF:
\begin{lemma}
\label{lem:dd-ef}
Given item rankings $\succ_1,\ldots,\succ_n$:

(a) An allocation $\bf{X}$ is NDDEF iff $\forall i,j\in\allagents:~  X_i~\succsim_i^{NDD}~X_j$.

(b) If an allocation $\bf{X}$ is PDDEF, then $\forall i,j\in\allagents:~  X_i~\succsim_i^{PDD}~X_j$.
\end{lemma}

\begin{proof}
Let $EF(i,j,\mathbf{u})$ be the no-envy predicate $u_i(X_i) ~\geq~ u_i(X_j)$.

(a)
The NDDEF definition is
``For all DD utility profiles $\mathbf{u}$, for all $i$ and for all $j$, $EF(i,j,\mathbf{u})$.''
The right-hand side is 
``for all $i$, for all $j$, for all DD utility profiles $\mathbf{u}$, $EF(i,j,\mathbf{u})$.''
Switching the order of for-all quantifiers yields logically-equivalent statements.

(b) 
The PDDEF definition is 
``there exists $\mathbf{u}$ for which, for all $i$ and for all $j$, $EF(i,j,\mathbf{u})$.''
The right-hand side is ``for all $i$ and for all $j$, there exists $u_i$ by which $EF(i,j,\mathbf{u})$''. 
The former statement logically implies the latter.%
\footnote{
In contrast to Lemma \ref{lem:dd-prop},
here the latter statement (which can be called ``Weak-PDDEF'') does not imply the former (PDDEF) when there are three or more agents. For example, if $X_1 \succsim^{PDD}_1 X_2$ and $X_1 \succsim^{PDD}_1 X_3$, then it is possible that agent 1 does not envy agent 2 by some
DD function $u_{1,2}$, and does not envy agent 3 by some other DD function $u_{1,3}$, but there is no \emph{single} DD function by which agent 1 envies neither agent 2 nor agent 3.
}
\end{proof}
Based on Lemma \ref{lem:dd-ef}(a), Corollary \ref{cor:polynomial} extends to NDDEF: it is possible to decide in polynomial time whether a given allocation is NDDEF.
However, we do not have a strongly-polynomial time algorithm for deciding whether a given allocation is PDDEF.%
\footnote{
The situation is similar for PosEF, see \citet{Aziz2015Fair}. 
For both fairness criteria, deciding whether a given allocation is fair can be done using a linear program with $m n$ variables describing the ``witness'' utility profile. The constraints require that the allocation is fair, and (for PDDEF) also that the utility profile satisfies the DD condition. This requires weakly-polynomial time.
As far as we know, it is an open question whether the decision problem can be solved in strongly-polynomial time.
}
Below we focus on the NDDEF criterion.

Since every NDDEF allocation is NDDPR, the two conditions of Theorem \ref{thm:nddpr} are necessary for the existence of NDDEF allocations for any number of agents. In the special case of $n=2$ agents, NDDPR is equivalent to NDDEF so these conditions are also sufficient. But for $n\geq 3$ they are no longer sufficient.

\begin{example}
\label{exm:nddef-3}
There are six items $\{1,\ldots,6\}$. The preferences of the three agents Alice Bob and Carl are:
\begin{itemize}
\item[] Alice: $6 \succ 5 \succ 3 \succ 4 \succ 2\succ 1$
\item[] Bob:  $5 \succ 4\succ 3\succ 6\succ 2\succ 1$ 
\item[] Carl:  $4\succ 6\succ 3\succ 5\succ 2\succ 1$
\end{itemize}
The conditions of Theorem \ref{thm:nddpr} are clearly satisfied: the number of items is a multiple of 3 and the best items are all different.
However, no NDDEF allocation exists.
To see this, note that the preferences are the same up to a cyclic permutation between $6$ $5$ and $4$, so the agents are symmetric and it is without loss of generality to assume that Alice receives item 1. Therefore, to ensure proportionality, Alice's bundle must be $\{6,1\}$ and her Borda score is $7$. To ensure that Alice is not envious,
both Bob and Carl must get items with a Borda score (for Alice) of $7$. Thus there are two cases:

(a) Bob gets $\{5,2\}$ and Carl gets $\{3,4\}$. This allocation is NDDPR but it is not NDDEF, since Bob envies Carl according to the Borda score. 

(b) Carl gets $\{5,2\}$ and Bob gets $\{3,4\}$. This allocation is not even NDDPR since Carl's Borda score is 5 (and Carl necessarily envies Bob).
\qed

\end{example}

When the number of agents is not bounded, deciding the existence of NDDEF allocations is computationally hard:
\begin{theorem}
\label{thm:nddef-npcomplete}
When there are $n\geq 3$ agents and at least $2 n$ items, checking the existence of NDDEF allocations is NP-complete (as a function of $n$).
\end{theorem}
\begin{proof}[Proof Sketch]
By Lemma \ref{lem:dd-ef}, to check whether an allocation is NDDEF, we have to do at most $n^2$ checks of the $\succsim^{NDD}$ relation. Each such check can be done in polynomial time by Theorem \ref{thm:X-NDD-Y} and Algorithm \ref{algo:comparedd}. Hence the problem is in NP. 
The proof of NP-hardness is similar to the proof of \citet{Bouveret2010Fair} for the NP-hardness of checking the existence of necessarily\-/envy-free allocations. The proof requires carefully checking that the reduction argument works for NDDEF as well.
The details are presented in Appendix\ref{app:proofs}.
\end{proof}

When the number of agents is \emph{constant} (at least 3) and the number of items is variable, the runtime complexity of checking NDDEF  existence is an open question: is it polynomial in $m$ like NDDPR, or NP-hard like necessary\-/EF \citep{Aziz2016Control}?

\section{Pareto\-/efficiency}
 \label{sec:efficiency}
An allocation is called \emph{Pareto\-/efficient} if every other allocation is either not better for any agent, or worse for at least one agent:
\begin{definition}[Pareto\-/efficiency]
\label{def:dd-pe}
Given utility functions $u_1,\ldots,u_n$, an allocation $\bf{X}$ is called \emph{Pareto\-/efficient (PE)} if
for every other allocation $\bf{Y}$,
either 
$\forall i\in\allagents: u_i(X_i) \geq u_i(Y_i)$, or 
$\exists i\in\allagents: u_i(X_i) > u_i(Y_i)$.
~~~
Based on this definition, \emph{necessary\-/DD-Pareto\-/efficiency (NDDPE)} and 
\emph{possible\-/DD-Pareto\-/efficiency (PDDPE)}
are defined analogously to NDDPR and PDDPR.

\end{definition}
The criteria of \emph{necessary\-/Pareto\-/efficiency (NecPE)}
and \emph{possible\-/Pareto\-/efficiency (PosPE)} are defined analogously.
It is clear from the definition that NecPE implies NDDPE implies PDDPE implies PosPE.
With the analogous fairness criteria, these implications are strict, i.e., some possibly-fair allocations are not PDD-fair, and some NDD-fair allocations are not necessarily\-/fair.
But with Pareto\-/efficiency the situation is different: 
\begin{theorem}
An allocation is NecPE if and only if it is NDDPE.
\end{theorem}
\begin{proof}
The implication NecPE $\implies$ NDDPE is obvious by the definition. 
We now consider an allocation $\bf{X}$ that is not NecPE and prove $\bf{X}$ is not NDDPE.

By \citet{Aziz2016Optimal} Theorem 9,
if $\bf{X}$ is not NecPE then
 there are two options:

(i) $\bf{X}$ is not possibly-PE. Then, it is certainly not NDD-PE.

(ii) $\bf{X}$ admits a \emph{Pareto\-/improving one-for-two-swap}. This means that there are two agents, say Alice and Bob, such that $X_A$ contains an item $x$, 
$X_B$ contains two items $y,z$, 
and Bob strictly prefers the one item over each of the two: $x \succ_B y$ and $x\succ_B z$.
Then $X$ is not NDD-PE, since it is not PE for the following utilities:
\begin{align*}
u_A(x) = m^2+\level_A(x)
&&
u_B(x) = 2^{\level_B(x)}
\end{align*}
Note that both utility functions have DD.
Alice's utility is dominated by the number of items she has, so she always prefers two items to one.
Bob's utility is lexicographic, so he always prefers one good item to any number of worse items.
Hence, by switching $\{x\}$ and $\{y,z\}$
we get a new allocation that is strictly better for both Alice and Bob, and does not affect any other agent.
\end{proof}
\begin{theorem}
An allocation is PosPE if and only if it is PDDPE.
\end{theorem}
\begin{proof}

The implication PDDPE $\implies$ PosPE is obvious by definition. 
We now consider an allocation $\bf{X}$ that is not PDDPE and prove $\bf{X}$ is not PosPE.

Consider the \emph{lexicographic} utility profile, by which for each $i\in\allagents$, $u_i(x) = 2^{\level_i(x)}$. Since these utilities have DD, $\bf{X}$ is not PE according to this profile. So there exists an allocation $\bf{Y}$ 
by which for some agent, say Alice: $u_A(Y_A)>u_A(X_A)$, and for all agents $B$: $u_B(Y_B)\geq u_B(X_B)$.

Since Alice prefers $Y_A$ to $X_A$ by a lexicographic utility function, there exists some integer $k\geq 1$ such that $X_A$ and $Y_A$ contain the same $k-1$ best items, but the $k$-th best item in $Y_A$ (denoted by $y_a$) is better for Alice than the $k$-th best element in $X_A$.

In allocation $\bf{X}$, item $y_a$ belonged to some other agent, say Bob.  
But Bob must be weakly better-off in $\bf{Y}$ than in $\bf{X}$,
so $Y_B$ must contain a better item that was not in $X_B$; 
let's call this item $y_b$.

In allocation $\bf{X}$, item $y_b$ belonged to some other agent,
say Carl. 
From similar considerations, Carl must have in $\bf{Y}$ an item
$y_c$ that he prefers to $y_b$.
Continuing this way, we end with a cycle of agents, 
each of whom gave an item to the previous agent and received a
\emph{better} item from the next agent.

Now consider the allocation $\bf{Z}$ which is identical to
$\bf{X}$ except that the single-item exchanges in the cycle take place (so $y_a$ is given to Alice, $y_b$ is given to Bob and so on).
Then $\bf{Z}$ is better than $X$ for all agents in the cycle,
and this is true for any additive utility function.
Hence, $\bf{X}$ is not possibly-PE.
\end{proof}

So DD leads to new fairness criteria but not to new efficiency criteria.

\section{Conclusions and Future Work}
We formalized natural ways to compare sets of goods by using the DD (diminishing differences) assumption. 
In Appendix \ref{sec:increasing}
we present the analogous ID (increasing differences) assumption for chores.
The relations lead to new fairness criteria which we studied in detail. 
Two main open questions remain for future work: 
one about envy-free allocation of goods (Section \ref{sec:NDDEF}),
and one about fair allocation of chores (Appendix \ref{sec:increasing}). Below we present the smallest cases in which these questions are open.
\begin{enumerate}
\item There are three agents with different rankings over $m$ goods.
Can it be decided in time polynomial in $m$, whether there exists a
necessary\-/DD envy-free allocation?
\item There are three agents with different rankings over $m$ chores.
Can it be decided in time polynomial in $m$, whether there exists a
necessary\-/ID proportional allocation?
\end{enumerate}

Besides these questions, 
it may be interesting to extend the results to the case where agents may be indifferent between items.%
\footnote{
Theorem \ref{thm:X-NDD-Y} is proved for multi-bundles so it holds with indifferences too. But Theorem \ref{thm:nddpr} fails. 
Consider an instance with $m=8$ goods and  $n=2$ agents with the following rankings:
\begin{align*}
Alice:&&   a\succ b\succ c\succ d
 = w = x = y \succ z
\\
Bob:&&   b\succ a\succ c\succ d
 = w = x = y \succ z
\end{align*}
The agents have different best goods, so we might think that balanced-round-robin might yield an NDDPR allocation. However, when goods are picked in the order ABBAABBA, Alice's bundle is $\{a,d,x,z\}$; it is not NDDPR for her, since it is not proportional by Borda scores (the total Borda score is $5+4+3+2+2+2+2+1=21$, while Alice's Borda score is $5+2+2+1=10$).
}

Additionally, it may be interesting to identify other interesting set extensions that correspond to classes of utility functions. 
For example, suppose that agents care both about getting a best item and about not getting a worst item, but do not care much about intermediate items (so the differences in utilities are decreasing at first and then increasing). What can be said of fair allocations under this assumption?

\section*{Acknowledgments}
We acknowledge the Dagstuhl Seminar 16232 on Fair Division where this project was initiated.
We are grateful to four anonymous IJCAI reviewers and three anonymous JAIR reviewers for their very helpful comments.

Haris Aziz is supported by a Scientia Fellowship.
Erel Segal-Halevi was supported by the ISF grant 1083/13, the Doctoral Fellowships of Excellence Program and the Mordecai and Monique Katz Graduate Fellowship Program at Bar-Ilan University. Avinatan Hassidim is supported by ISF grant 1394/16.

\newpage
\appendix

\section{Chores and Increasing Differences}
\label{sec:increasing}
In this section, we assume that we have to divide indivisible \emph{chores}, defined as items with negative utilities.
Therefore, all the utility functions we consider in this section assign strictly negative values to all items. 

With chores, the Diminishing Differences condition means that the difference between the easiest to the second-easiest chore is larger than the difference between the second-hardest to the hardest chore. But usually, with chores, people care more about not getting the hardest chores than about getting the easiest chores. 
Therefore, we introduce the condition of  \emph{increasing} differences (ID).
In many aspects, the ID condition for chores is analogous to the DD condition for goods (subsection \ref{sub:id-basic}). However, finding necessarily\-/ID-fair allocation for chores is more difficult than necessarily\-/DD-fair allocation for goods (subsections \ref{sub:id-fair},\ref{sub:id-2},\ref{sub:id-3}).
\subsection{Increasing differences: basic definitions}
\label{sub:id-basic}
The following definition is analogous to Definition \ref{def:dd-utilities}:
\begin{definition}
\label{def:id-utilities}
Let $\succ$ be a preference relation and
$u$ a utility function consistent with $\succ$.
We say that
$u$ has the \emph{Increasing Differences (ID)} property if, for every three items with consecutive levels $x_3\succ x_2\succ x_1$ such that $\level(x_3)=\level(x_2)+1=\level(x_1)+2$,
it holds that $u(x_3)-u(x_2)\leq u(x_2)-u(x_1)$.

We denote by $\uid(\pref)$ the set of all ID utility functions consistent with $\pref$.

Given $n$ rankings $\pref_1,\ldots,\pref_n$, 
We denote by $\uid(\pref_1,\ldots,\pref_n)$ the set of all vectors of ID utility functions, $u_1,\ldots,u_n$, such that $u_i$ is consistent with $\pref_i$.
\end{definition}
There is a one-to-one correspondence between DD utilities and ID utilities.
Given a strict ranking $\succ$, define its reverse ranking $\rev{\succ}$ as:

\begin{align*}
\forall x,y\in\allitems: &&
y \rev{\succ} x\iff x\succ y 
\end{align*}
Given a utility function $u$, define its reverse function $\rev{u}$ as:
\begin{align*}
\forall x\in\allitems:&&\rev{u}(x) := - u(x)
\end{align*}
\begin{lemma}
\label{lem:dd-id}
For every ranking $\succ$ and utility function $u$:
\begin{align*}
\rev{u} \in \uid(\rev{\succ})
\iff
u\in\udd(\succ).
\end{align*}
\end{lemma}
\ifdefined\FULLVERSION
\begin{proof}
Clearly 
$\rev{u}$ is consistent with $\rev{\succ}$,
iff 
$u$ is consistent with $\succ$. Now:

$\rev{u} \in \uid(\rev{\succ})$ ~~~$\iff$

for every three consecutive items $x_3\rev{\succ} x_2\rev{\succ} x_1$:\\
$~~~~~~~~\rev{u}(x_3)-\rev{u}(x_2)\leq \rev{u}(x_2)-\rev{u}(x_1)$ ~~~$\iff$

for every three consecutive items $x_1\succ x_2\succ x_3$:
\\
$~~~~~~~~ [-u(x_3)]- [-u(x_2)]\leq [-u(x_2)] - [-u(x_1)]$ ~~~$\iff$

for every three consecutive items $x_1\succ x_2\succ x_3$:
\\
$~~~~~~~~~u(x_1)-u(x_2) \geq u(x_2) - u(x_3)$ ~~~$\iff$

$u\in \udd(\succ)$.
\end{proof}
\else
The proof is technical and we omit it.
\fi

The negative-Borda utility function, $u_{-Borda}(x) := \level(x)-m-1$, is a member of $\uid(\succ)$, as well as the negative-lexicographic utility function,
$u_{-Lex}(x) := - 2^{m-\level(x)}$. By the latter function, the bundles are first ranked by whether they contain the worst chore, then by whether they contain the next-worst chore, etc.

An alternative characterization of $\uid$ is given by the following lemma. It is analogous to Lemma \ref{lem:u-in-DD} and proved in a similar way, so we omit the proof:
\begin{lemma}
\label{lem:u-in-ID}
$u\in \uid(\pref)$ iff, for every four items $x_2, y_2, x_1, y_1$ with $x_2\succsim x_1$ and $y_2\succsim y_1$
and $x_2\neq y_2$ and $x_1\neq y_1$:
\begin{align*}
\frac{u(x_2)-u(y_2)}{\level(x_2)-\level(y_2)}
\leq
\frac{u(x_1)-u(y_1)}{\level(x_1)-\level(y_1)}
\end{align*}
\end{lemma}

Analogously to Definition \ref{def:dd-relation} we define the relations $X \succsim^{NID} Y$ and $X \succsim^{PID} Y$.
These are closely related to their DD counterparts:
\begin{lemma}
\label{lem:ndd-nid}
Let $\succ$ be a ranking and $\rev{\succ}$ its inverse ranking.  Then, for every two multi-bundles $X,Y$:
\begin{align*}
X \succsim^{NID} Y
&&
\iff
&&
Y \rev{\succsim}^{NDD} X
\\
X \succsim^{PID} Y
&&
\iff
&&
Y \rev{\succsim}^{PDD} X
\end{align*}
\end{lemma}
\ifdefined\FULLVERSION
\begin{proof}
By Lemma \ref{lem:dd-id}:

$X \succsim^{NID} Y$ ~~~$\iff$

$\forall u\in \uid(\pref) \text{:~~~~~} u(X)\geq u(Y)$ ~~~$\iff$

$\forall \rev{u}\in \udd(\rev{\pref}) \text{:~~~~~} \rev{u}(X)\leq \rev{u}(Y)$ ~~~$\iff$

$Y \rev{\succsim}^{NDD} X$.

\noindent
Similarly:

$X \succsim^{PID} Y$ ~~~$\iff$

$\exists u\in \uid(\pref) \text{:~~~~~} u(X)\geq u(Y)$ ~~~$\iff$

$\exists \rev{u}\in \udd(\rev{\pref}) \text{:~~~~~} \rev{u}(X)\leq \rev{u}(Y)$ ~~~$\iff$

$Y \rev{\succsim}^{PDD} X$.
\end{proof}
\else
Again the proof is technical and we omit it.
\fi

Thus, to check whether $X \succsim^{NID} Y$ / $X \succsim^{PID} Y$ with regards to some ranking $\succ$, 
we can simply use Algorithm $\ref{algo:comparedd}$ with the inverse ranking $\rev{\succ}$.

We now want to prove an analogue of Theorem \ref{thm:X-NDD-Y} for chores.
For this, we order the chores in each multi-bundle by \emph{increasing} level, so  $X=\{x_1,\ldots, x_{|X|}\}$ where $x_1\preceq_i \ldots \preceq_i x_{|X|}$
For each $k\leq |X|$ we define $X^{k}$ as the $k$ \emph{worst} chores in $X$, $X^{k}:=\{x_1,\ldots,x_k\}$.

\begin{theorem}
\label{thm:X-NID-Y}
Given a ranking $\succ$ and two (multi-)bundles $X,Y$ of chores, 
$X\succsim^{NID} Y$ if and only if both of the following conditions hold:
\begin{enumerate}
\item $|X|\leq |Y|$;
\item For each $k\in \{1,\ldots, |Y|\}$:
$\level(X^{k}) \geq  \level(Y^{k})$.
\end{enumerate}
\end{theorem}
Note that condition (i) is the opposite of condition (i) in Theorem \ref{thm:X-NDD-Y}: $X$ must have weakly \emph{less} chores than $Y$. However, condition (ii) is identical to condition (ii) in Theorem \ref{thm:X-NDD-Y}.
\begin{proof}
Define the \emph{inverse-level} of an item/bundle as its level under the inverse-ranking $\rev{\succsim}$. So the inverse-level of the hardest chore is $m$ and of the easiest chore is $1$.

By Lemma \ref{lem:ndd-nid}, 
$X \succsim^{NID} Y$ iff
$Y \rev{\succsim}^{NDD} X$.
By Theorem \ref{thm:X-NDD-Y}, this holds iff both the following conditions hold:
\begin{enumerate}
\item $|Y|\geq |X|$;
\item For each $k\in \{1,\ldots, |Y|\}$,
the inverse-level of the $k$ chores in $Y$ that are best by $\rev{\succ}$ (i.e., worst by $\succ$), is at least as high as the inverse-level of the 
$k$ chores in $X$ that are worst by $\succ$.
\end{enumerate}
The first condition is equivalent to $|X|\leq |Y|$ and the second condition is equivalent to $\level(X^{k}) \geq  \level(Y^{k})$.
\end{proof}

\subsection{Increasing differences: fairness criteria}
\label{sub:id-fair}
Analogously to Definition \ref{def:dd-prop}, 
we define the fairness criteria NIDPR (necessary\-/ID-proportional) and PIDPR (possible\-/ID-proportional). 
Analogously to Lemma \ref{lem:dd-prop} and Corollary \ref{cor:polynomial}, and with similar proofs that we omit, we have:
\begin{lemma}
\label{lem:id-prop}
Given item rankings $\succ_1,\ldots,\succ_n$:
\begin{itemize}
\item An allocation $\bf{X}$ is NIDPR iff $\forall i\in\allagents:~ n\cdot X_i~\succsim_i^{NID}~\allitems$.
\item An allocation $\bf{X}$ is PIDPR iff
$\forall i\in\allagents:~ n\cdot X_i~\succsim_i^{PID}~\allitems$.
\end{itemize}
\end{lemma}
\begin{corollary}
\label{cor:polynomial-id}
The following problems can be decided in polynomial time:

(a) Given an allocation, decide whether it is NIDPR;

(b) Given an allocation, decide whether it is PIDPR.
\end{corollary}

In Section \ref{sec:NDDPR} we proved that an NDD-proportional allocation exists whenever the number of items is an integer multiple of the number of agents, and all agents have different best items. 
At first glance, the natural extension of this condition to chores is that all agents should have different worst chores. The following two examples show that this condition is neither sufficient nor necessary.
\begin{example}
\label{exm:nid-not-sufficient}
There are eight chores and four agents with rankings:
\begin{align*}
A:&&   a\succ b\succ c\succ d
\succ w\succ x\succ y\succ z
\\
B:&&   b\succ c\succ d\succ a
\succ w\succ x\succ z\succ y
\\
C:&&  c\succ d\succ a\succ b
\succ w\succ z\succ y\succ x
\\
D:&&   d\succ a\succ b\succ c
\succ x\succ z\succ y\succ w
\end{align*}
Each agent has a different best chore and each agent has a different worst chore.
However, at least one agent (the one who receives $y$) has a second-worst chore. This implies that an NIDPR allocation does not exist. To see this, suppose that all agents have the same ID scoring function:
\begin{align*}
-996,-997,-998,-999,
-1000,-2000,-3000,-4000
\end{align*}
The utility of the agent who receives $y$ is at most $-3996$. However, the total value is $-13990$ and 
the fair share is $-13990/4 = -3497.5$. \qed
\end{example}

\begin{example}
\label{exm:nid-not-necessary}
There are three chores and three agents with rankings:
\begin{align*}
A:&&   x\succ y\succ z
\\
B:&&   x\succ z\succ y
\\
C:&&  x\succ z\succ y
\end{align*}
All agents have the same best chore, and two agents have the same worst chore. However, the following allocation is NIDPR: 
\begin{align*}
A: \{y\} && B: \{x\} && C: \{z\}
\end{align*}
This is obvious for Bob since he receives his best (easiest) chore. To see that it is also true for Alice, we show that $3\cdot X_A \succsim^{NID}_A \allitems$ using Theorem \ref{thm:X-NID-Y}. Condition (i) clearly holds since both multi-bundles have 3 chores. For Condition (ii), compare the levels of the $k$ worst chores, for $k=1,2,3$:
\begin{center}
\begin{tabular}{llll}
& $k=1$ & $k=2$ & $k=3$ \\
$3\cdot X_A$	& $2$ & $2$ & $2$ \\
$\allitems$		& $1$ & $2$ & $3$ \\
Difference		& $+1$ & $0$ & $-1$ \\
Accumulated difference & $+1$ & $+1$ & $0$
\end{tabular}
\end{center}
The accumulated difference is always at least 0, so $3\cdot X_A \succsim^{NID}_A \allitems$. 
By a similar calculation, $3\cdot X_C \succsim^{NID}_C \allitems$. Hence the allocation is NIDPR.\qed
\end{example}

Below we present a different condition that is necessary for the existence of NIDPR allocations.
It is analogous to the ``only-if'' part of Theorem \ref{thm:nddpr}.
To state this condition, for each agent $i$, let $W_i$ be the set of $i$'s $\lceil{n-1\over 2}\rceil$ worst chores.

\begin{theorem}
\label{thm:nidpr-necessary}
If there exists a NIDPR allocation of chores among $n$ agents, then both the following conditions must hold: 

(a) The number of chores is $m = l\cdot n$, for some integer $l$.

(b) It is possible to allocate to each agent $i$, $l$ chores that are not from $W_i$. \\
(Hence, the intersection of all  $\lceil{n-1\over 2}\rceil$-worst-chores sets is empty: $\cap_{i\in\allagents} W_i = \emptyset$).
\end{theorem}
\begin{proof}
Let $(X_1,\ldots,X_n)$ be an NIDPR allocation. Then for every agent $i$,\\ $n\cdot X_i ~\succsim^{NID}_i~ \allitems$. By Theorem \ref{thm:X-NID-Y}.

(a) For every $i\in\allagents$: $|n\cdot X_i| \leq  |\allitems|
\implies n\cdot |X_i|  \leq  m.$
But this must be an equality since the total number of items in all $n$ bundles is exactly $m$. So the total number of items is $n\cdot |X_i|$ which is an integer multiple of $n$.

(b) For every $i\in\allagents$, the level of the $n$ worst chores in $n\cdot X_i$ must be weakly larger than the level of the $n$ worst chores in $\allitems$. 
The $n$ worst chores in $\allitems$ have levels $1,\ldots,n$, so their total level is ${n(n+1)\over 2}$.
The $n$ worst chores in $n\cdot X_i$ are just $n$ copies of the worst chore in $X_i$.
Thus, the level of this chore must be at least 
${n(n+1)\over 2} / n = {n+1\over 2}$. Since levels are integers, the smallest level in $X_i$ must be at least $\lceil {n+1\over 2} \rceil$.
So the agent must not get 
any of his $\lceil {n-1\over 2} \rceil$ worst chores. In other words, agent $i$ must not get any chore from the set $W_i$.
Since all chores must be allocated, no chore may be in the intersection of all $W_i$.
\end{proof}
In Example \ref{exm:nid-not-sufficient}, $\lceil {n-1\over 2}\rceil = 2$, and the intersection of the 2-worst-chores sets is not empty (it contains chore $y$), so a NIDPR allocation does not exist.
In Example \ref{exm:nid-not-necessary}, 
 $\lceil {n-1\over 2}\rceil = 1$, the intersection of the worst-chore sets is empty (not all three agents have the same worst chore), and a NIDPR allocation exists.

We do not know if the condition of Theorem \ref{thm:nidpr-necessary} is sufficient for the existence of NIDPR allocations in general. 
Below we prove that they are sufficient in two special cases: two agents, and three agents with ``almost'' identical rankings.

\subsection{NIDPR allocation for two agents}
\label{sub:id-2}
With two agents, for each $i\in\{1,2\}$, the set $W_i$ contains just the worst chore of agent $i$, so the necessary condition of Theorem \ref{thm:nidpr-necessary} simply says that each agent has a different worst chore.
This condition is also sufficient for the existence of NIDPR allocations. 
The following theorem is analogous to the ``if'' part of Theorem \ref{thm:nddpr} for $n=2$.
\begin{theorem}
\label{thm:nidpr-sufficient-2}
There exists a NIDPR allocation of chores among $n=2$ agents whenever the following conditions both hold: 

(a) The number of chores is $m = l\cdot n$, for some integer $l$.

(b) The worst chores of the agents are different.

In case it exists, it can be found in time $O(m)$.
\end{theorem}

\ifdefined\DirectProof
\begin{proof}
We prove that the balanced-round-robin algorithm (Algorithm \ref{algo:roundrobin}), used in the proof of Theorem \ref{thm:nddpr}, produces an NIDPR allocation in our case.
Let $X_i$ be the bundle allocated to agent $i$ by this algorithm. 
We prove that $n\cdot X_i \succsim^{NID}_i \allitems$ by the two conditions of Theorem \ref{thm:X-NID-Y}.

Condition (i) is satisfied with equality, since by (a) each agent gets exactly $l$ items, so $|n\cdot X_i| = n l = m = |\allitems|$.

Condition (ii) says that, for every $k\range{1}{m}$, the total level of the $k$ worst chores in the multi-bundle $n\cdot X_i$ is at least as large as the total level of the $k$ worst chores in $\allitems$.
In other words, the level-difference $\level_i((n\cdot X_i)^{k}) - \level_i(\allitems)$ must be at least 0 for every $k$.
Note that each agent receives chores from best to worst. We are interested in levels of chores from worst to best, so we analyze the algorithm from the last round towards the first.
We check that the total level-difference (the variable TotalLevelDiff in Algorithm \ref{algo:comparedd}) is always at least 0.

The worst chore in $X_i$ is the chore given to agent $i$ in the last round. 
In the last round, only two chores remain to be allocated. By condition (b) the worst chores of the agents are different, hence the best remaining chores of the agents are also different. Hence, in the last round each agent receives a chore with a level of at least 2, so the level of the worst $2$ chores in $n\cdot X_i$ is at least $2$. 
The following table shows the levels and their differences for $k\in\{1,2\}$:
\begin{center}
\begin{tabular}{lll}
 	& $k=1$ & $k=2$ \\
$n\cdot X_i$	& $\geq 2$ & $\geq 2$ \\
$\allitems$		& $1$ & $2$ \\
LevelDiff		& $\geq 1$ & $\geq 0$ \\
TotalLevelDiff	& $\geq 1$ & $\geq 1$ \\
\end{tabular}
\end{center}
Clearly the level difference is weakly-positive and there is a total accumulated difference of at least 1. 

We now prove that, for each round $r\geq 1$ (counting from the last one), the accumulated level-difference 
for $k\in\{2 r - 1, 2 r \}$ is at least $1$ when $r$ is odd, and at least $4 - 2 i$ when $r$ is even. In particular, it is always at least 0.
The proof is by induction on $r$. We have just proved the base $r=1$.

Suppose now that $r$ is even and $r>1$.
When agent $i$ gets a chore,
the number of remaining chores is 
$r n + 1 - i$, so agent $i$'s best remaining chore has a level of at least $r n + 1 - i = 2 r + 1 - i$. 
Therefore, the level-differences 
for $k\in\{2 r - 1, 2 r \}$ are as in the following table (the last row takes into account an accumulated difference of at least $+1$ from the induction assumption):
\begin{center}
\begin{tabular}{lll}
	$n\cdot X_i$	& $\geq 2 r + 1 - i$ & $\geq 2 r + 1 - i$ \\
	$\allitems$		& $2 r - 1$ & $2 r$ \\
	LevelDiff		& $\geq 2 - i$ & $ \geq 1 - i$ \\
	TotalLevelDiff	& $\geq 3-i$ & $\geq 4 - 2 i$ \\
\end{tabular}
\end{center}
So for both $i$, the accumulated differences are positive and their total is at least $4 - 2 i$, as claimed.

Suppose now that $r$ is odd and $r>1$.
When agent $i$ gets a chore, 
the number of remaining chores is $rn + i - 2$. Therefore, agent $i$'s best remaining item has a level of at least $rn + i - 2 = 2 r + i - 2$. Therefore, the level-differences 
for $k\in\{2 r-1,r\}$ are as in the following table(the last row takes into account an accumulated difference of at least $4 - 2 i$ from the induction assumption):
\begin{center}
	\begin{tabular}{lll}
		$n\cdot X_i$	& $\geq 2 r + i - 2$ & $\geq 2 r + i - 2$ \\
		$\allitems$		& $2 r - 1$ & $2 r$ \\
		LevelDiff		& $\geq i-1$ & $\geq i-2$ \\
		TotalLevelDiff	& $\geq 3-i$ & $\geq 1$ \\
	\end{tabular}
\end{center}
So for both $i$, the accumulated differences are positive and their total is at least $1$, as claimed.
\end{proof}
\else
Theorem \ref{thm:nidpr-sufficient-2} can be proved directly by analyzing the outcome of the balanced round-robin protocol (Algorithm \ref{algo:roundrobin}), similarly to the proof of Theorem \ref{thm:nddpr}.
This analysis is technical and we omit it.

Intuitively, when there are two agents, allocating chores is equivalent to allocating exemptions from chores.
An exemption from chore is a good;
hence, chore allocation is equivalent to good allocation.%
\footnote{
This observation was already made by \citet{Bogomolnaia2017Competitive} for divisible resources,
and proved formally by 
\citet{segal2018competitive}
for competitive equilibrium with indivisible objects.
}
An exemption from the worst chore is the best good; hence, Theorem \ref{thm:nddpr} implies Theorem \ref{thm:nidpr-sufficient-2}.%
\footnote{
The round-robin protocol would be slightly different in case of chores:  each agent should pick an \emph{exemption} from a chore, rather than a chore. In other words, each agent in turn should pick a chore and give it to the \emph{other} agent.
}
\fi

\subsection{NIDPR allocations for three agents}
\label{sub:id-3}
The analogy between goods and chores does not extend to $n\geq 3$ agents.%
\footnote{
as already noted by  \citet{Bogomolnaia2017Competitive} for divisible resources.
}
This is because for each chore, there are $n-1$ identical exemptions to share, and each agent must get at most one  such exemption; this constraint does not exist in the problem of allocating goods.

Hence, Theorem \ref{thm:nidpr-sufficient-2} does not generalize to three or more agents. 
The balanced-round-robin protocol does not necessarily find a NIDPR allocation, even if it exists. In Example \ref{exm:nid-not-necessary}, the rankings satisfy the necessary condition of Theorem \ref{thm:nidpr-necessary}, and a NIDPR allocation exists, but the round-robin protocol (in the order A B C) yields the allocation:
\begin{align*}
A: \{x\} && B: \{z\} && C: \{y\}
\end{align*}
which is not NIDPR since it gives Carl his worst chore.

For three agents, we consider the following special case:
\begin{itemize}
\item All agents have the same $n$ worst chores;
\item All agents have the same $m-n$ best chores, and rank them identically.
\end{itemize}
In some sense this is a ``worst case'' of fair allocations, since the agents' preferences are as similar as they can be without violating the necessary condition. 

We prove that, in this ``worst case'', the necessary condition of Theorem \ref{thm:nidpr-necessary} is also sufficient.
\begin{theorem}
\label{thm:nidpr-sufficient-3}
There exists a NIDPR allocation of chores among $n=3$ agents whenever the following conditions hold: 

(a) The number of chores is $m = l\cdot n$, for some integer $l$.

(b) Not all agents have the same worst chore;

(c) All agents have the same $n$ worst chores;

(d) All agents have the same $m-n$ worst chores and rank them identically.

In this case, it can be found in time $O(m)$.
\end{theorem}
\begin{proof}
We first allocate the $n$ worst chores.
By condition (b), it is possible to give each agent a chore 
with a level of at least 2. Moreover, by simple case analysis it is possible to see that it is always possible to give at least one agent a chore with a level of at least 3. Hence, after this step, 
the total level-differences of all agents are at least 0:
\begin{center}
\begin{tabular}{llll}
 	& $k=1$ & $k=2$ & $k=3$\\
$n\cdot X_i$	& $2$ & $2$ & $2$ \\
$\allitems$		& $1$ & $2$ & $3$ \\
LevelDiff		& $1$ & $0$ & $-1$ \\
TotalLevelDiff	& $1$ & $1$ & $0$\\
\end{tabular}
\end{center}
and the total level-difference of at least one agent is $3$:
\begin{center}
\begin{tabular}{llll}
 	& $k=1$ & $k=2$ & $k=3$\\
$n\cdot X_i$	& $3$ & $3$ & $3$ \\
$\allitems$		& $1$ & $2$ & $3$ \\
LevelDiff		& $2$ & $1$ & $0$ \\
TotalLevelDiff	& $2$ & $3$ & $3$\\
\end{tabular}
\end{center}
We now have $m-n$ remaining chores. By condition (d), the levels of these chores are the same for all agents, namely, $4,\ldots, m$.
We allocate them from worst ($4$) to best ($m$), using a round-robin protocol. There are $l-1$ allocation rounds; in each round, the first (worst) chore is given to an agent whose TotalLevelDiff is at least 3. We prove by induction that, indeed, when each round ends, there is at least one agent with TotalLevelDiff at least 3, while all other agents have TotalLevelDiff at least 0.

The induction base ($r=1$) was already proved above. Assume the claim is true until the beginning of some round $r$. The level of the next chore to allocate is $3 r - 2$. It is given to an agent with TotalLevelDiff at least 3, so his levels change as follows:
\begin{center}
\begin{tabular}{llll}
	& $k=3 r - 2$ & $k=3 r - 1$ & $k=3 r$\\
$n\cdot X_i$	& $3 r - 2$ & $3 r - 2$ & $3 r - 2$ \\
 $\allitems$	& $3 r - 2$ & $3 r - 1$ & $3 r$ \\
 LevelDiff		& $0$ & $-1$ & $-2$ \\
 TotalLevelDiff	& $\geq 3$ & $\geq 2$ & $\geq 0$\\
\end{tabular}
\end{center}
The next chore is $3 r - 1$. It is given to an agent with TotalLevelDiff at least 0, so his levels change as follows:
\begin{center}
\begin{tabular}{llll}
	& $k=3 r - 2$ & $k=3 r - 1$ & $k=3 r$\\
$n\cdot X_i$	& $3 r - 1$ & $3 r - 1$ & $3 r - 1$ \\
 $\allitems$	& $3 r - 2$ & $3 r - 1$ & $3 r$ \\
 LevelDiff		& $1$ & $0$ & $-1$ \\
 TotalLevelDiff	& $\geq 1$ & $\geq 1$ & $\geq 0$\\
\end{tabular}
\end{center}
The next chore is $3 r$. It is given to an agent with TotalLevelDiff at least 0, so his levels change as follows:
\begin{center}
\begin{tabular}{llll}
	& $k=3 r - 2$ & $k=3 r - 1$ & $k=3 r$\\
$n\cdot X_i$	& $3 r$ & $3 r$ & $3 r$ \\
 $\allitems$	& $3 r - 2$ & $3 r - 1$ & $3 r$ \\
 LevelDiff		& $2$ & $1$ & $0$ \\
 TotalLevelDiff	& $\geq 2$ & $\geq 3$ & $\geq 3$\\
\end{tabular}
\end{center}
As claimed, after round $r$ ends, all agents have TotalLevelDiff at least 0, and one agent has TotalLevelDiff at least 3.

Hence, the resulting allocation is NIDPR.
\end{proof}
Theorem \ref{thm:nidpr-sufficient-3} can be extended to more than 3 agents. Whenever the worst $n$ chores can be allocated such that the total level-difference of all agents is at least 0 and the total level-difference of some agents is sufficiently high, it is possible to allocated the other chores such that the total level-difference of all agents remains at least 0. 
Moreover, instead of requiring that all agents have exactly the same ranking to their $m-n$ best chores, it is sufficient that all agents have the same worst $n$ chores (levels $1,\ldots,n$), the same next-worst $n$ chores (levels $n+1,\ldots,2n$), etc. 
We omit these results since we believe that the main interesting challenge is generalizing the theorem to arbitrary rankings. Finding a general sufficient condition and protocol for NIDPR allocation of chores remains an interesting open problem.

\section{Binary Utilities}
\label{sec:binary}
In this section, we compare the diminishing/increasing differences assumptions to another natural assumption, which we call \emph{Binary}.
It is based on the assumption that each agent only cares about getting as many as possible of his $k$ best items, where $k$ is an integer that may be different for different agents. 
The binary assumption was also studied by \citet[Proposition 21]{bouveret2008efficiency}, who proved that finding an efficient envy-free allocations with such preferences is NP-complete. 

The following definition is analogous to Definitions \ref{def:dd-utilities} and \ref{def:id-utilities}:
\begin{definition}
\label{def:bin-utility}
Let $\succ$ be a preference relation and
$u$ a utility function consistent with $\succ$.
We say that
$u$ is \emph{Binary} if,
for some integer $k\geq 1$:
\begin{align*}
u(x) = 
\begin{cases}
1 & \text{when~} \level(x)\geq k
\\
0 & \text{when~} \level(x)< k
\end{cases}
\end{align*}
We denote by $\ubin(\pref)$ 
the set of all binary utility functions consistent with $\pref$.
\end{definition} 

Analogously to Definition \ref{def:dd-relation} we define the relations $X \succsim^{NBIN} Y$ and $X \succsim^{PBIN} Y$;
analogously to Definition \ref{def:dd-prop} 
we define NBIN-fairness and PBIN-fairness.

At first glance, the Binary assumption seems much more restrictive than the DD assumption. For every $\pref$, the set $\ubin(\pref)$ contains only $m$ utility functions --- much less than $\udd(\pref)$. Therefore, one could expect NBIN-fairness to be easier to satisfy than NDD-fairness. But this is not the case:  NBIN-fairness is equivalent to necessary fairness and PBIN-fairness is equivalent to possible fairness. This follows from the following theorem.
\begin{theorem}
For every item-ranking $\succ$ and every multi-bundles $X,Y$:

(a) $X \succsim^{Nec} Y$ if and only if $X \succsim^{NBIN} Y$ and

(b) $X \succsim^{Pos} Y$ if and only if $X \succsim^{PBIN} Y$.
\end{theorem}
\begin{proof}
It is sufficient to prove the following directions:

(a) If $X \succsim^{NBIN} Y$ then $X \succsim^{Nec} Y$;

(b) If not $X \succsim^{PBIN} Y$ then not $X \succsim^{Pos} Y$.

For the proof, we use the following notation.
\begin{itemize}
\item The $m$ items are denoted by their level, so the best item is $m$ and the worst is $1$.
\item For a multi-bundle $X$ and an item $j$, the number of copies of $j$ in $X$ is denoted $X[j]$.
\item The $m$ utility functions in 
$\ubin(\succ)$ are denoted by $U_k$, for $k\range {1}{m}$.
\end{itemize}
In this notation, for every $k\range {1}{m}$ and multi-bundle $X$:
\begin{align}
\label{eq:U_k}
U_k(X) = 
\sum_{j=k}^m X[j]
\end{align}
so $U_m(X)=X[m]$ (the agent cares only about the best item), $U_{m-1}(X)=X[m]+X[m-1]$ (the agent cares only about the two best items), etc.

Moreover, for every function $u\in \uall(\succ)$ and multi-bundle $X$:
\begin{align}
\label{eq:u}
u(X) &= 
\sum_{j=1}^m u(j)\cdot X[j]
\end{align}
Substituting the $X[j]$ in \eqref{eq:u} using \eqref{eq:U_k} gives:
\begin{align*}
u(X) &= 
u(m)\cdot U_m(X) + \sum_{j=2}^{m}
 u(j-1)\cdot \bigg(U_{j-1}(X)-U_{j}(X) 
\bigg)
\\
&= 
\sum_{j=2}^{m}
\bigg(u(j) - u(j-1)\bigg)\cdot U_j(X)
+
u(1) \cdot U_1(X)
\end{align*}
so every additive function $u$ is a linear combination of the functions $U_k$. Note that all coefficients in this linear combination are non-negative. Hence:

\begin{itemize}
\item If $\forall k\range{1}{m}$: $U_k(X)\geq U_k(Y)$, then $\forall u\in\mathcal{U}(\succ): u(X)\geq u(Y)$. This implies (a).
\item If $\forall k\range{1}{m}$: $U_k(X) < U_k(Y)$, then $\forall u\in\mathcal{U}(\succ): u(X)<u(Y)$. This implies (b).
\qedhere
\end{itemize}
\end{proof}

\section{NP-hardness of NDDEF}
\label{app:proofs}

\noindent
\textbf{Theorem \ref{thm:nddef-npcomplete}.} 
When there are $n\geq 3$ agents and at least $2 n$ items, checking the existence of NDDEF allocations is NP-hard (as a function of $n$).
\begin{proof}
The proof is similar to the proof of \citet{Bouveret2010Fair} for the NP-hardness of checking existence of NecEF allocations.
We now present their reduction and show that it works for NDDEF as well.

The proof is by reduction from the \emph{exact-3-cover} problem, whose inputs are:
\begin{itemize}
\item A base set of $3 q$ elements;
\item A set-family containing $n\geq q$ triplets, $C_1,\ldots,C_n$, each of which contains exactly 3 elements from the base-set.
\end{itemize}
The question is whether there exist $q$ pairwise-disjoint triplets whose union is the base-set.
Given an instance of exact-3-cover, an instance of fair item allocation is constructed as follows:
\begin{itemize}
\item To the $3 q$ base elements correspond $3 q$ \emph{main items}, denoted by $Main$. 
To each triplet $C_i$ corresponds a set of three main items, denoted by $Main_i$, such that $\forall i\range{1}{n}: Main_i\subseteq Main$. The sets $Main_i$, like the triplets $C_i$, are \emph{not} necessarily disjoint. We denote by $Main_{-i}$ the main items not in $Main_i$.
\item There are also $3 n$ \emph{dummy items} denoted by $Dummy$. 
To each triplet $C_i$ corresponds a set of three dummy items, denoted by $Dummy_i:=\{d_{i},d_{i'},d_{i''}\}$. All such sets are pairwise disjoint. We denote by $Dummy_{-i}$ the dummy items not in $Dummy_i$.
\item There are $3 (n-q)$ \emph{auxiliary items}, denoted by $Aux$. They are partitioned to $n-q$ pairwise-disjoint triplets, denoted by $Aux_j:=\{x_{j},x_{j'},x_{j''}\}$, for $j\range{q+1}{n}$.
All in all, there are $6 n$ items.
\item To each triplet $C_i$ corresponds a set of three \emph{agents}, $Agents_i = \{i,i',i''\}$. The sets $Agents_i$ are pairwise disjoint. All in all, there are $3 n$ agents.
\item The preferences of the three agents in $Agents_i$ are, in general:
\begin{align*}
Dummy_i \succ Main_i \succ Aux_{q+1} \succ \cdots \succ Aux_{n} \succ Dummy_{-i} \succ Main_{-i}
\end{align*}
Their preferences over the three items in $Dummy_i$ are ``cyclic'', i.e., for agent $i$ it is $d_{i}\succ d_{i'}\succ d_{i''}$, for agent $i'$ it is $d_{i'}\succ d_{i''}\succ d_{i}$, and for agent $i''$ it is $d_{i''}\succ d_{i}\succ d_{i'}$. Their preferences over the three items in $Main_i$ are cyclic in a similar way.
Their preferences over the three items in $Aux_j$, for each $j\range{q+1}{n}$, are cyclic in a similar way.
Their preferences over $Dummy_{-i}$ and $Main_{-i}$ are arbitrary.
\end{itemize}

\citet{Bouveret2010Fair} prove that there exists a NecEF allocation iff there exists an exact-3-cover. The proof involves three arguments:
\begin{enumerate}
	\item \label{item:nef} In a NecEF allocation, each agent must receive the same number of items. Here there are $6 n$ items and $3 n$ agents so each agent must get exactly two items. One of these items must be its top dummy item, which is easy to do since the top dummy items of all agents are different. So, it remains to prove that there is an exact-3-cover, if and only if the second items can be allocated in a NecEF way, i.e., such that each agent prefers the worst item in his bundle to the worst item in any other bundle.
	\item  \label{item:cover-nef} $Cover \implies allocation$: Suppose there is an exact-3-cover, e.g, with the triplets $C_1,\ldots,C_q$. Then, the sets of main items $Main_1,\ldots,Main_q$ are pairwise-disjoint and their union is exactly $Main$. Then, for each $j\range{1}{q}$, it is possible to allocate the three items in $Main_j$ to the three agents in  $Agents_j$, giving each agent his favorite main item.
	Let's call these $3 q$ agents in the triplets $Agents_1,\ldots,Agents_q$, the ``lucky agents''.
	The allocation is NecEF for the lucky agents since their worst item is their 4th-best item while the worst item in any other bundle is at most their 5th-best item (since their three best items are the dummy items and they are already allocated).
	It remains to determine an allocation for the $3(n-q)$ ``unlucky'' agents, $Agents_{q+1},\cdots,Agents_n$. For each $j\range{q+1}{n}$, give to the three agents in $Agents_j$, the three items in $Aux_j$, giving each agent his favorite item from that triplet. 
	This item is better for them than the worst items in the other bundles, which are from  $Main_{-i}$ or $Dummy_{-i}$, so the allocation is NecEF for them too.
	\item  \label{item:nef-cover} $Allocation \implies cover$: Suppose there is a NecEF allocation. For each $i\range{1}{n}$, consider the three agents in $Agents_i$. We claim that either all of them receive a main item from $Main_i$, or none of them does. \emph{Proof:} suppose e.g. that agent $i$ receives item $m_i\in Main_i$ but agent $i'$ does not receive any item from $Main_i$. Then, the allocation of $i$ is $\{d_i,m_i\}$ and the best possible allocation for $i'$ is $\{d_{i'},x_{i'}\}$, where $x_{i'}$ is the auxiliary item preferred by agent $i'$. But for agent $i'$, both items allocated to agent $i$ are better than $x_{i'}$. Therefore agent $i'$ might envy $i$, so the division is not NecEF.
	Since each main item must be allocated to exactly one agent, there exists an exact-3-cover: the triplet $C_i$ is in the cover if and only if the agents in $Agents_i$ receive the items in $Main_i$.
\end{enumerate}
We now show that the reduction also works for NDDEF. Claim \ref{item:nef} works for NDDEF by Theorem \ref{thm:X-NDD-Y}. Claim \ref{item:cover-nef} clearly works for NDDEF since every NEF allocation is NDDEF. 
It remains to prove claim \ref{item:nef-cover}. 
Suppose there exists an NDDEF allocation.
This allocation is, in particular, envy-free according to the Borda score. 
We claim that, for each 
$i\range{1}{n}$, either all three agents in $Agents_i$ receive an item from $Main_i$, or none of them does.
\emph{Proof:} consider the following two cases:
\begin{itemize}
	\item One agent, say $i$, receives his main item $m_i$, but the other agents $i',i''$ do not receive their main items. The dummy items give agent $i''$ a Borda-advantage of 1 over $i$. The best second item that can be allocated to $i''$ is his best auxiliary item, but this leaves him with a Borda-disadvantage of 2 relative to $i$, so $i''$ Borda-envies $i$.
	\item Two agents, say $i',i''$, receive their main items $m_{i'},m_{i''}$, but agent $i$ does not receive his main item. The dummy items give agent $i$ a Borda-advantage of 1 over $i'$. The best second item that can be allocated to $i$ is his best auxiliary item, but this leaves him with a Borda-disadvantage of 2 relative to $i'$, so $i$ Borda-envies $i'$.
\end{itemize}
Therefore the reduction is valid for NDDEF too, and the theorem is proved.
\end{proof}

\newpage
\bibliographystyle{elsarticle-harv}
\bibliography{merged}

\end{document}